\documentclass{theoretics}

\usepackage{amsfonts}
\usepackage{textcomp}
\usepackage{xspace}
\usepackage[english]{babel}
\usepackage[utf8]{inputenc}
\usepackage{csquotes} %Uri
\usepackage[classfont = bold]{complexity}
\usepackage{mathtools}
\usepackage{bussproofs}
\usepackage{subfig}
\usepackage{todonotes}
\usepackage{thm-restate}
\usepackage[capitalise, noabbrev, nameinlink]{cleveref}

\newcommand{\newauthor}[3]{
    \newcounter{#1comment}
    \setcounter{#1comment}{1}
    \expandafter\newcommand\csname #1\endcsname[1]{%
            \par\noindent
            \todo[inline, size = \small, backgroundcolor = {#3}, caption = {}]{
                \arabic{#1comment}:
                {##1} --~\textbf{#2}
            }
            \addtocounter{#1comment}{1}
    }
    \expandafter\newcommand\csname #1changed\endcsname[1]{%
        \ifdraft
            \colorbox{#3}{
                ##1
            }%
        \else
            ##1%
        \fi
    }
}

\newcommand{\pa}[1]{}
\newcommand{\kr}[1]{}

\DeclareMathOperator{\Char}{char}
\newcommand{\eps}{\varepsilon}

\newcommand{\eqcomma}{\enspace ,}
\newcommand{\eqperiod}{\enspace .}
\newcommand{\disjointunion}{\overset{.}{\cup}}
\newcommand{\prob}[2][]{\Pr_{#1}[#2]}
\newcommand{\Prob}[2][]{\Pr_{#1} \bigl[ #2 \bigr]}
\providecommand{\abs}[1]{\lvert#1\rvert}
\providecommand{\Abs}[1]{\bigl\lvert#1\bigr\rvert}
\newcommand{\ABS}[1]{\left\lvert#1\right\rvert}
\newcommand{\set}[1]{\{#1\}}
\newcommand{\setdescr}[2]{\{\,#1 \mid #2\,\}}
\newcommand{\floor}[1]{\lfloor #1 \rfloor}
\newcommand{\ceil}[1]{\lceil #1 \rceil}
\newcommand\restrict[2]{{% we make the whole thing an ordinary symbol
  \left.\kern-\nulldelimiterspace % automatically resize the bar with \right
  #1 % the function
  \vphantom{\big|} % pretend it's a little taller at normal size
  \right|_{#2} % this is the delimiter
  }}
\usepackage{bbm}
\newcommand{\ind}[1]{\mathbbm{1}_{\set{#1}}}

\newcommand{\calB}{\mathcal{B}}

\newcommand{\calF}{\mathcal{F}}
\newcommand{\calG}{\mathcal{G}}
\newcommand{\calH}{\mathcal{H}}

\newcommand{\calP}{\mathcal{P}}
\newcommand{\calQ}{\mathcal{Q}}

\renewcommand{\E}{\mathbb{E}}
\renewcommand{\R}{\mathbb{R}}
\newcommand{\N}{\mathbb{N}}
\newcommand{\Z}{\mathbb{Z}}
\newcommand{\F}{\mathbb{F}}

\DeclareMathOperator{\Deg}{\sf{Deg}}
\DeclareMathOperator{\Size}{\sf{Size}}
\newcommand{\RefuteDeg}[2]{\Deg(#2 \vdash_{\text{\tiny #1}} \, \perp)}
\newcommand{\RefuteSize}[2]{\Size(#2 \vdash_{\text{\tiny #1}} \, \perp)}

\DeclareMathOperator{\Card}{Card}

\newcommand{\subformula}{$ is an affine restriction of $}

\usepackage{mathrsfs}
\newcommand{\frege}[1][]{$\mathscr{F}_{{#1}}$\xspace}

\newcommand{\psys}{$P$\xspace}

\newcommand{\pmf}[1]{\mathsf{PM}(#1)}
\newcommand{\pmax}[1]{q^{\mathsf{PM}}_{#1}}

\newcommand{\tsf}[1]{\tau(#1)}
\newcommand{\tsax}[1]{q^{\tau}_{#1}}

\newcommand{\pedef}{\tilde{\mathbb{E}}}
\newcommand{\pe}[2][]{\pedef_{#1}[#2]}

\newcommand{\vtex}[2][v]{(#1, #2)}

\newcommand{\lift}[1]{\mathsf{lift}(#1)}

\newcommand{\gts}{G}
\newcommand{\vts}{V(G)}
\newcommand{\ets}{E(G)}

\newcommand{\gpm}{H}

\newcommand{\epm}{E(H)}

\newcommand{\nonbip}[2]{$(#1, #2)$-odd-cycle-robust}

\newcommand{\degrich}[2]{$(#1, #2)$-max-degree-robust}
\newcommand{\DEGRICH}[2]{$\left(#1, #2\right)$-max-degree-robust}

\newcommand{\cdiam}[1][\alpha]{D^{\text{{\o}}}_{{#1}}}

\newcommand{\diam}[1]{\text{diam}(#1)}

\newcommand{\cross}[1]{$#1$-cross}

\newcommand{\embed}[1]{B_{#1}}

\newcommand{\sgood}{S^{\star}}

\newcommand{\totallength}[3][]{L^{{#1}}_{\text{disj}}(#2, #3)}

\newcommand{\branch}{U}

\newcommand{\branches}{\mathcal{U}}

\newcommand{\bFree}{\branches_{\text{free}}}

\DeclareMathOperator{\CSP}{CSP}

\title{Perfect Matching in Random Graphs
           is
           as Hard as Tseitin\TCSthanks{
           Supported by the
           Approximability and
           Proof Complexity
           project funded by the
           Knut and Alice Wallenberg
           Foundation.}}

\TCSauthor[affil1]{Per Austrin}{austrin@kth.se}
\TCSauthor[affil1]{Kilian Risse}{kilianr@kth.se}

\TCSaffil[affil1]{KTH Royal Institute of Technology}

\TCSvolume{2022}
\TCSarticlenum{2}
\TCSreceived{Jan 28, 2022}
\TCSrevised{Jul 28, 2022}
\TCSaccepted{Aug 3, 2022}
\TCSpublished{Dec 15, 2022}
\TCSkeywords{proof complexity, perfect matching, topological embedding, sum of squares, polynomial calculus, bounded depth Frege}
\TCSdoi{10.46298/theoretics.22.2}
\TCSshortnames{Per Austrin and Kilian Risse}
\TCSshorttitle{Perfect Matching in Random Graphs is as Hard as Tseitin}
\TCSthanks{A previous version of this paper has appeared in the
  proceedings of the 33rd Annual ACM-SIAM Symposium on Discrete Algorithms (SODA 2022).}

\begin{document}
\maketitle

\begin{abstract}
  We study the complexity of proving that a sparse random regular
  graph on an odd number of vertices does not have a perfect matching,
  and related problems involving each vertex being matched some
  pre-specified number of times.
  We show that this requires proofs of degree $\Omega(n / \log n)$ in
  the Polynomial Calculus (over fields of characteristic $\ne 2$) and
  Sum-of-Squares proof systems, and exponential size in the
  bounded-depth Frege proof system.
  This resolves a question by Razborov asking whether
  the Lovász-Schrijver proof system requires~$n^\delta$ rounds to
  refute these formulas for some $\delta > 0$.
  The results are obtained by a
  worst-case to average-case reduction of these formulas relying on a
  topological embedding theorem which may be of independent interest.
\end{abstract}

% \thispagestyle{empty}
% \clearpage

% \setcounter{page}{1}

%% \tableofcontents

\section{Introduction}

Proof complexity is the study of certificates of unsatisfiability,
initiated by Cook and Reckhow~\cite{cook79efficiency}
as a program to separate $\NP$
from $\coNP$. The main goal of this program is to prove size lower
bounds on proofs of unsatisfiability of logical formulas. This
is a daunting job -- indeed we are far
from proving general size lower bounds on certificates of
unsatisfiability.
As an intermediate step % to proving general lower bounds
we study proof
systems with restricted deductive power
and prove size lower bounds for such restricted certificates of
unsatisfiability.
The most studied such proof system is resolution~\cite{Blake37Thesis} which is fairly well understood by now, see
e.g., the proof complexity book by Krajíček
\cite{Krajicek19ProofComplexity}.

But resolution is by far not the only proof
system. A closely related and quite general proof system is the
bounded depth Frege proof system \cite{cook79efficiency} which
manipulates propositional formulas of bounded depth. While we
have some results for the bounded depth Frege proof system, in this
introduction we instead focus on two other systems as these
were the primary motivation behind our work.
These are the two proof systems
Polynomial Calculus (PC) \cite{CEI96Groebner, ABRW04PRG} and
Sum-of-Squares (SoS)
\cite{Shor87, parrilo2000, lasserre2001}.
These proof systems do not rely on propositional logic, like
resolution or Frege, but rather on
algebraic reasoning and are examples of
so-called (semi-)algebraic proof systems (see
e.g. \cite{GHA02proofs}).

Both PC and SoS provide refutations of (satisfiability of) a set of
polynomial equations $\calQ = \setdescr{q_i(x) = 0}{i \in [m]}$ over
$n$ variables $x_1, \ldots, x_n$.  In the case of PC, these
polynomials can be over any field $\F$ (finite or infinite), and in
the case of SoS, these polynomials are over $\R$.
A key complexity
measure of a PC$_\F$ or SoS refutation of $\calQ$ is its degree, defined as
the maximum degree of any polynomial appearing in the refutation.
The degree of refuting
$\calQ$ in PC$_\F$ or
SoS, which we denote by $\RefuteDeg{PC$_\F$}{\calQ}$ and
$\RefuteDeg{SoS}{\calQ}$ respectively, is the minimum degree of
any PC$_\F$ or SoS refutation of $\calQ$. For Boolean systems of
equations, meaning that $\calQ$ contains the equations
$x_i^2 - x_i = 0$ for all $i \in [n]$,
strong enough degree lower bounds imply size
lower bounds in both PC$_\F$ \cite{CEI96Groebner,IPS99LowerBounds}
and SoS \cite{AH18SosTradeoff}, where the size of a
refutation is the total number of monomials
appearing in it. For finite $\F$ the proof system
PC$_\F$ is incomparable 
to SoS \cite{Razborov98, gri01xor, GHA02proofs} whereas SoS
can simulate PC$_\R$ by the recent result of
Berkholz~\cite{Berkholz18Relation}.
%For formal definitions consult \cref{sec:pcplx}.

There is by now a large number of lower bound results for both
PC
\cite{Razborov98, IPS99LowerBounds, BGIP01LinearGaps, AR01nonbin,
  GL10Optimality, MN15},
and SoS
\cite{gri01xor, Schoenebeck08LinearLevel,
  MPW15SumOfSquaresPlantedClique, BHKKMP16clique,
  kmow17anycsp, AH18SosTradeoff,
  potechin20ordering, AGK20RefuteSemi},
with SoS in particular having received considerable
attention in recent years due to its close connection to the
Sum-of-Squares hierarchy of semidefinite programming,
a powerful ``meta-algorithm'' for
combinatorial optimization problems \cite{BS14SumsOfSquaresICM}.

In this paper we study the power (or lack thereof) of these proof
systems when it comes to refuting the perfect matching formula
$\pmf{G}$
defined over
sparse random graphs $G = (V, E)$ on an odd number of vertices.
This formula can be viewed as a system of linear equations over $\R$
on a set of Boolean variables: for each edge $e \in E$ there is
a variable $x_e \in \{0,1\}$ (indicating whether the edge is used in
the matching) and for each vertex $v \in V$ there is an equation
$\sum_{e \ni v} x_e = 1$.
Apart from being a natural well-studied problem on its own, the
perfect matching formula is interesting because of its close relation to two
other widely studied families of formulas, namely the pigeonhole principle (PHP), and
Tseitin formulas.

PHP asserts that $m$ pigeons cannot fit in $n < m$ holes (where each
hole can fit at most one pigeon). This can be viewed as a bipartite
matching problem on the complete bipartite graph with $m+n$ vertices,
where each vertex on the large side (with $m$ vertices) must be
matched at least once, and each vertex on the small side (with $n$
vertices) can be matched at most once.  There are many variants of
PHP (see e.g. the survey \cite{Razborov02ProofComplexityPHP}),
and the one closest to the perfect matching formula is the
so-called ``onto functional PHP'', in which each vertex on
both sides must be matched exactly once (rather than at least/at most
once).  Equivalently, this formula is simply the perfect matching formula
on a complete bipartite graph with $n+m$ vertices.
While most variants of PHP are hard for PC \cite{Razborov98, MN15}, the
onto functional PHP variant is in fact easy to refute % in PC
in PC over any field \cite{Riis93Thesis}. In SoS, all variants
of PHP are easy to refute \cite{GHA02proofs}.

The Tseitin formula over a graph $G$ claims that
there is a
subgraph of $G$ such that each vertex has odd degree. As the
sum of the degrees of a graph is even, this formula is not
satisfiable if $G$ has an odd number of vertices. In contrast
to the PHP, the Tseitin formula is (almost) always hard: for PC$_\F$
over fields $\F$ of characteristic distinct from 2
\cite{BGIP01LinearGaps, AR01nonbin}
and SoS \cite{gri01xor}
these formulas require linear
degree if $G$ is a good vertex expander.
We cannot hope to prove degree lower bounds over
fields of characteristic 2 as the constraints become linear
and we can thus refute
the Tseitin formula using Gaussian elimination.
As the perfect matching formula $\pmf{G}$ implies the Tseitin formula,
PC over fields of characteristic $2$ can also easily refute $\pmf{G}$
for $G$ with an odd number of vertices.

In summary,
the perfect matching formula lies somewhere in between PHP and
Tseitin, of which the former is easy to refute in SoS (and easy to
refute in PC in the onto functional variant), and the latter is hard
to refute in SoS (as well as in PC with characteristic $\ne 2$).
Hence it is natural to wonder whether SoS or PC requires large degree to
refute the perfect matching formula over non-bipartite graphs.

The case of
perfect matching in the \emph{complete graph} on an odd number of vertices
(sometimes called the ``MOD 2 principle'') is well-understood in both
PC \cite{BGIP01LinearGaps} and SoS \cite{gri01xor,potechin17GoodStory},
requiring degree $\Omega(n)$ in both proof systems unless the
underlying field of PC is of characteristic 2.
For sparse graphs, less is known.  Buss
et al.~\cite{BGIP01LinearGaps} obtained worst-case lower bounds in PC
showing that there exist bounded degree graphs on $n$ vertices
requiring $\Omega(n)$ degree refutations. This is obtained by a
reduction from Tseitin formulas and while the work of Buss et
al.~predates the current interest in the SoS system, it is not hard to
see that the same reduction yields a similar $\Omega(n)$ degree lower
bound for SoS (details provided in \cref{sec:worst-case}).

However, for random graphs $G$ little is known
about the hardness of the perfect matching formula and, e.g.,
Razborov \cite{razborov2017width}
asked whether it is true that the Lovász-Schrijver hierarchy~\cite{LS91Cones} (which is weaker than SoS) requires
$n^{\eps}$ rounds to refute the perfect matching principle on a random
sparse regular graph with high probability.

\subsection{Our results}

We show that indeed the perfect matching principle requires large size
on random $d$-regular graphs (for some constant $d$) in the
Sum-of-Squares, Polynomial Calculus, and bounded-depth Frege proof
systems.  Our results apply more generally to Tseitin-like formulas defined by
linear equations over the reals induced by some graph, so let us now define these.

For a graph $G = (V,E)$ and integer vector $b \in \Z^V$, consider the
system of linear equations over the reals having a variable $x_e$ for
each $e \in E$, and the equation $\sum_{e \ni v} x_e = b_v$ for each
$v \in V$.  Let $\Card(G, b)$ denote this system of linear equations
along with the Boolean constraints $x_e \in \{0,1\}$ (viewed as a
quadratic equation $x_e^2-x_e = 0$) for each edge -- in
\cref{sec:formulas} the encoding is discussed in more detail.  Note
that $\Card(G, \vec{1})$ corresponds to the perfect matching problem
in $G$ and in general $\Card(G, b)$ can be viewed as asserting that
$G$ has a ``matching'' where each vertex is matched exactly $b_v$
times.  Note that whenever $\sum_{v \in V} b_v$ is odd, $\Card(G, b)$
is unsatisfiable (since the equations imply
$\sum_{v} b_v = 2\sum_{e} x_e$ which is even)\footnote{As pointed out
  to us by Aleksa Stanković, decidability of $\Card(G, b)$ is in
  polynomial time: starting with the all $0$ assignment, iteratively
  build up an assignment that may match some vertices fewer times than
  required. If there is a satisfying assignment, then there is always
  an augmenting path along which the current assignment can be
  improved, i.e., more edges set to $1$, by a similar argument as for
  matchings \cite{Berge57}. Such a path can be found in polynomial
  time by an adaptation of the blossom algorithm \cite{edmonds65}.}.

We focus on the special case of $\Card(G, b)$ where $G$ is $d$-regular
and $b = \vec{t} = (t, t, \ldots, t)$ is the all-$t$ vector for some
$t \in [d]$.  If in this scenario both $n$ and $t$ are odd (implying
$d$ is even) then as observed above $\Card(G, \vec{t})$ is
unsatisfiable.  On the other hand if $n$ is odd and $t$ is even then
$\Card(G, \vec{t})$ is always satisfiable (because such $G$ admits a
$2$-factorization).  The remaining case when $n$ is even may be either
satisfiable or unsatisfiable, but for a random $d$-regular $G$ with
$d \ge 3$, $\Card(G, t)$ will be satisfiable with high probability
(because such $G$ can be partitioned into perfect matchings with high
probability).

If we let \frege[D] denote a Frege system restricted to depth-$D$ formulas
(see \cref{sec:pcplx}),
then our main theorem is as follows.

\begin{theorem}
  \label{thm:pm avg lower bound}
  There is a constant $d_0$ such that for all constants $d \ge d_0$ and
  $t \in [d]$, the following holds asymptotically almost surely
  over a random $d$-regular graph $G$ on $n$ vertices.
  \begin{enumerate}
  \item $\RefuteDeg{PC$_\F$}{\Card(G, \vec{t})} = \Omega(n/\log n)$
    for any fixed field $\F$ with $\Char(\F) \ne 2$.
  \item $\RefuteDeg{SoS}{\Card(G, \vec{t})} = \Omega(n/\log n)$.
  \item There is a $\delta > 0$ such that
    $\RefuteSize{\frege[D]}{\Card(G, \vec{t})} =
    \exp\big(\Omega(n^{\delta/D})\big)$,
    for all
    $D \le \frac{\delta \log n}{\log\log n}$.
  \end{enumerate}
\end{theorem}

The interesting case of the above theorem is when both $n$ and $t$ are
odd so that $\Card(G, t)$ is unsatisfiable; in the other cases
$\Card(G, \vec{t})$ is satisfiable with high probability and the
lower bounds are vacuous.

By known size-degree tradeoffs for Polynomial Calculus
\cite{IPS99LowerBounds, CEI96Groebner} and Sum-of-Squares
\cite{AH18SosTradeoff} the degree
lower bounds in \cref{thm:pm avg lower bound} imply near-optimal size
lower bounds of $\exp\big(\Omega(n/\log^2 n)\big)$.

Apart from the perfect matching formula, another special case of
$\Card(G, \vec{t})$ is the so-called even coloring formula,
introduced by Markström \cite{Markstrom06Locality},
which is the case when $t = \deg(v)/2$.
An open problem of Buss and Nordström
\cite[Open Problem 7.7]{BN20ProofCplx}
asks whether these formulas
are hard on spectral expanders for Polynomial Calculus over fields of
characteristic $\ne 2$. \cref{thm:pm avg lower bound} partially resolves
this open problem, establishing that it is hard on random graphs
(rather than on all spectral expanders). See \cref{sec:conclusion} for
some further remarks on what parts of our proof use the randomness
assumption.

We will give a more detailed overview of how the results are obtained
in \cref{sec:proof overview} below, but for now let us mention that we obtain them
using embedding techniques, as introduced to proof complexity by
Pitassi et al.~\cite{pitassi16frege} (see discussion of related work
in \cref{sec:related work}).
In particular for, say, the SoS lower
bound, our starting point is the $\Omega(n)$ \emph{worst-case} degree
lower bound in sparse graphs, and we then prove that these hard
instances can be embedded in a random $d$-regular graph in such a way
that the hardness of refuting the formula is preserved.

To achieve this, one of the components we need is a new graph
embedding theorem which may be of independent interest.
Very loosely speaking,
we show that any bounded-degree graph with $O(n/\log n)$
edges can be embedded as a \emph{topological minor} in any
bounded-degree $\alpha$-expander on $n$ vertices and sufficiently many
edges.
In addition, for our application to perfect matching (and more
generally the $\Card(G, \vec{t})$ formulas), we need to be able to
control the parities of the path lengths used in the topological
embedding, and we show that as long as every large linear-sized subgraph
contains an odd cycle of length $\Omega(1/\alpha)$, this is indeed
possible.

Somewhat informally, we prove the following.

\begin{theorem}[Informal statement of \cref{thm:odd-minor}]
  \label{thm:odd-minor-informal}
  Let $G$ be a
  constant degree
  $\alpha$-expander on
  $n$ vertices.
  If $H$ is a graph with at most $\frac{\eps n}{\log n}$
  edges and
  $\Delta(H) \ll \alpha^2 \cdot d(G)$,
  then $G$ contains
  $H$ as a topological minor.
  Furthermore, if all large vertex induced subgraphs of $G$
  contain an odd cycle of length $\Omega(1/\alpha)$,
  then one can choose the parities
  of the length of all the edge embeddings in the minor.
\end{theorem}

This generalizes various classical results of a similar flavor
(e.g.~\cite{kr96ShortPaths,kn2018MinorsNoCut,CR2019large,krivelevich2018expanders}).
See the next subsection for a discussion comparing these (and other) existing
embedding results to ours.

As a further illustration of the applicability of this theorem we
partially resolve a question of
Filmus et al.~\cite{FLMNV13TowardsUnderstandingPC}.
They prove
that with high probability for random $d$-regular graphs $G$, where $d \ge 4$,
PC requires \emph{space}
$\Omega(\sqrt{n})$ to refute the Tseitin
formula, and conjecture that PC
in fact requires space $\Omega(n)$. On the other hand, Galesi et
al.~\cite{GKT19SpaceWidth} considered it plausible that the
$\Omega(\sqrt{n})$ bound is optimal. We (almost) resolve this question by
proving $\Omega(n /\log n)$ space lower bounds for the Tseitin formula
defined on vertex expanders, but only of large enough (constant)
average degree.

\begin{restatable}{theorem}{SpaceTheorem}\label{thm:space}
  For all $\alpha > 0$ there is a $d_0$ such that the following
  holds. Let $G$ be a bounded degree $\alpha$-expander on $n$
  vertices of average degree at least $d_0$. Then
  over any field $\F$ it holds that PC$_\F$ requires space
  $\Omega(n/\log n)$ to refute the Tseitin formula defined on $G$.
\end{restatable}

Let us mention that the constant hidden in the lower bound
$\Omega(n/\log n)$ depends on the maximum degree of $G$. Unlike
\cref{thm:pm avg lower bound}, vertex expansion is sufficient and we
require no randomness.  This lower bound is obtained by embedding a
worst-case instance, due to Filmus et al., into a vertex expander. We
provide more details in \cref{sec:space}.

\subsection{Related work}
\label{sec:related work}

\paragraph{Proof Complexity Lower Bounds Using Embedding Techniques}

There are a few other papers that employ embedding techniques in proof
complexity \cite{pitassi16frege, GI19, GIRS19FregeTseitinGeneral,
  IRSS19}, though none of these use the embedding techniques in
connection with algebraic systems like PC or SoS.  As far as we are aware
the first such work is that of Pitassi et
al.~\cite{pitassi16frege}, who apply embedding techniques to obtain
Tseitin lower bounds for Frege Systems, and their use is most similar to ours.
They rely on a
result of Kleinberg and Rubinfeld \cite{kr96ShortPaths} that
guarantees that any small enough graph is a minor of an expander (note
that we require \emph{topological} minors).  This is in contrast to
the other results that rely on the fundamental result that a graph of
large enough treewidth contains the grid graph as a minor \cite{RS86}.

\paragraph{Connection to Constraint Satisfaction Problems}

For a $k$-ary predicate $P: \{0,1\}^k \rightarrow \{0,1\}$, an
instance of the $\CSP(P)$ problem consists of a set of
constraints over $n$ Boolean variables $x_1, \ldots, x_n$, each
constraint being an application of $P$ on a list of $k$ variables.
The $\Card(G, \vec{t})$ formulas we study can be viewed as
instances of $\CSP(P)$ where each variable appears in exactly two
constraints and $P: \{0,1\}^d \rightarrow \{0,1\}$ is the constraint
that exactly $t$ of the $d$ inputs are $1$.

CSP problems have been extensively studied throughout the years, and
fairly general conditions under which $\CSP(P)$ is hard for PC and SoS
are known \cite{AR01nonbin,kmow17anycsp}. To be more accurate, these
results are for the more general $\CSP(P^{\pm})$ problem in which
each constraint is an application of $P$ on $k$ \emph{literals}
rather than variables.
In particular, Alekhnovich and Razborov \cite{AR01nonbin} showed that if $P$
is, say,
$8$-immune\footnote{$P$ is $r$-immune over $\F$ if there is no
degree-$r$ polynomial $q: \set{0,1}^k \to \F$
such that for all satisfying assignments
$\alpha \in \set{0,1}^k$ of $P$ it
holds that $q(\alpha) = 0$.} over the
underlying field $\F$, then
any PC$_\F$ refutation of
a random $\CSP(P^{\pm})$ instance with a
linear number of constraints
requires
degree $\tilde{\Omega}(n)$.
For SoS, Kothari et al.~\cite{kmow17anycsp}
showed that, if there
exists a pairwise
uniform distribution\footnote{A distribution $\mu$ over
$\{0,1\}^k$ is said to be pairwise uniform if for all $1 \le i < j \le
k$, the marginal distribution of $\mu$ restricted to coordinates $i$
and $j$ is uniform.}
$\mu$ over $\{0,1\}^k$ supported
on satisfying assignments of $P$, then with high probability a random
$\CSP(P^{\pm})$ instance on $m = \Delta n$ constraints needs degree
$\tilde{\Omega}(n/\Delta^2)$ to be refuted by the SoS proof system.

The predicates we study are linear equations over $\R$ and are neither
immune nor do they support a pairwise uniform distribution. As such,
our results provide CSP lower bounds that fall outside the
immunity and pairwise independence frameworks, which are the
source of a majority of existing CSP lower bounds in PC and SoS.  To
the authors' best knowledge the only other attempt to overcome this
framework in the average-case setting is the paper by Deshpande et
al.~\cite{DMOSS19}, showing lower bounds for the basic SDP of random
regular instances of $\CSP(\operatorname{NAE}_3^{\pm})$, where
$\operatorname{NAE}_3$ is the not-all-equal predicate on three bits.
In contrast to their work we show (almost) linear degree lower bounds
for the stronger Sum-of-Squares hierarchy, but only for a very wide
predicate of some large (but constant) arity.

\paragraph{Embedding Theorems}

There is a rich literature on embeddings of graphs as minors or
topological minors into expander graphs. We focus here on the ones
most closely related to \cref{thm:odd-minor-informal}.

The classical result of Kleinberg and Rubinfeld \cite{kr96ShortPaths} shows
that a regular expander $G$ on $n$ vertices contains every
graph $H$ with $O(n/\polylog(n))$ vertices and edges as a minor.
Krivelevich and Nenadov~\cite{kn2018MinorsNoCut} simplified and strengthened this by
improving the bound on the
size of $H$ to $O(n/\log n)$.  These results differ from ours in two
key ways: (i) we want topological minors, and (ii) we want to be able
to control the parities of the path lengths in the embedding.  We now
discuss these two aspects separately.

Results on topological minors, while somewhat less common, also
exist.  A result similar to ours is the result of Broder et
al.~\cite{BFSU96} that with high probability the random graph $\calG(n, m)$ on $n$
vertices and $m = \Omega(n \log n)$ edges contains any graph $H$ with
$\Delta(H) = O(m/n)$ and at most $O(n/\log n)$ edges (and at most
$n/2$ vertices) as a topological minor.

For our second property, the possibility to choose the
parities of the paths used in the topological embedding, we are not
aware of any previous work studying this question.
A related notion are so called \emph{odd minors} which are more general
than topological minors with odd length
paths. This notion has been considered in connection with a
strengthening of Hadwiger's Conjecture, see e.g., the survey by
Seymour~\cite{Seymour2016Hadwiger}.
This line of research mostly considers complete odd minors, e.g.,
\cite{GGRSV09OddMinor}, and thus is not directly applicable to our situation.

Recently Draganić et al.~\cite{DKN20Embed} independently obtained
a new embedding theorem similar to ours.
They assume the somewhat stronger property that the
host graph $G$ is a spectral expander but also obtain a stronger
conclusion:
each path of the topological embedding is of equal (odd) length and the
embedding even works in an adversarial setting.
Namely, the adversary
is allowed to fix the embedding of the vertices, as long as no
neighborhood in $G$ contains too many vertex embeddings.

The
embedding theorem of Draganić et al. can be used to implement our
proof strategy.
The results are unaffected by this change except
in the setting of \cref{thm:space}. There, instead of considering
vertex expanders, we need to consider regular spectral expanders
with the benefit that the required average degree $d_0$ is
considerably decreased.

\paragraph{Extended Formulations}

There has been a
fair amount of work studying the
\emph{extension complexity} of the
perfect matching polytope
\cite{yannakakis88, Rothvoss17},
but these lower bounds do not
have any direct implications for the
PC and SoS degree of the perfect
matching formula. Let us elaborate.

Suppose we have a convex polytope $\calP$ consisting of
many facets. A natural question is whether there is
simpler polytope $\calQ$ in a higher dimensional space
so that $\calP$ is the ``shadow'' of~$\calQ$, or a bit more formally that
there is a
linear projection $\pi$ such that $\pi(\calQ) = \calP$. Such a~$\calQ$
is then called a linear extension of $\calP$ and the
extension complexity of a polytope $\calP$ is the minimum number
of facets of any linear extension of $\calP$.

Rothvoss \cite{Rothvoss17} proved that the perfect matching polytope
of a complete $n$-node graph has extension compexity $\exp(\Omega(n))$
for $n$ even. This result is incomparable to our lower bounds: as the
graphs we consider do not contain a perfect matching, their perfect
matching polytope is empty and thus has extension complexity 0.
Rather than linear programs, i.e., polytopes, we consider semidefinite
programs which are more expressive. The extension complexity in the
semidefinite setting has also been studied before
\cite{LRS15,BBHPRRWZ17} but these results are incomparable for the
same reason just mentioned. While these results are incomparable,
it is worth mentioning that there is a connection between
Sherali-Adams (a proof system weaker than SoS) and extended
formulations \cite{CLRS16, KMR17}.

\subsection{Overview of Proof Techniques}
\label{sec:proof overview}

As previously mentioned, our high level approach is to first obtain
worst-case perfect matching lower bounds and to then embed these into
the $\Card(G, \vec{t})$ formula for $G$ a random regular graph.
The worst-case lower bounds are
obtained by a gadget reduction from Tseitin to perfect
matching, due to Buss et al.~\cite{BGIP01LinearGaps}. Using known
lower bounds for the Tseitin formula in the corresponding proof
systems
\cite{BGIP01LinearGaps, gri01xor, Has20} we
then obtain the desired worst-case lower bounds
for the perfect matching formula.

A naïve attempt to obtain average-case lower bounds from a sparse
worst-case instance $H$ on $n$ vertices is to topologically embed the
worst-case instance into a random regular graph $G$ on $O(n \log n)$
vertices using \cref{thm:odd-minor-informal}. One would then like to
argue that $\pmf{G}$ is hard.

Suppose each path $p_{uv}$ in the embedding of $H$ in $G$
corresponding to some edge $\set{u,v} \in E(H)$ is of odd length. Then
it is straightforward to verify that the perfect matching formula
defined over the embedding is at least as hard to refute as the
worst-case instance $\pmf{H}$: map each variable $y_e$, for
$e \in p_{uv}$, alternatingly to $x_{uv}$ or $\bar x_{uv}$ such that
the first and last edges of $p_{uv}$ are mapped to $x_{uv}$ (using
that $p_{uv}$ is of odd length). This simple projection maps the
perfect matching formula defined over the embedding of $H$ to
$\pmf{H}$ and thus shows that the hardness of $\pmf{H}$ should be
inherited.

But having such a
worst-case instance as a topological minor is \emph{not} sufficient to
conclude that $\pmf{G}$ is hard.
For instance $G$ may contain an isolated vertex
and it is then trivial to refute $\pmf{G}$. On the other hand
if we could
guarantee that there is a perfect matching $m$
in the subgraph of $G$ induced by the vertices \emph{not}
used in the embedding of $H$, we can conclude that $\pmf{G}$ is hard:
hit the formula with the restriction corresponding to the matching $m$
and by the argument from the previous paragraph we are
basically left with the worst-case formula.

Thus if we can ensure that $H$ is a topological minor of $G$
%% Our general recipe to obtain average-case lower bounds from these
%% sparse worst-case instances goes as follows. Start with the worst-case
%% instance $H$ on $n$ vertices and then show that with high probability,
%% $H$ can be appropriately embedded into a random regular
%% graph $G$ on $O(n \log n)$ vertices.
%% %% Start
%% %% with the sparse worst-case instance $H$ on $n$ vertices for perfect
%% %% matching from \cref{sec:worst-case} and then show that with high
%% %% probability, $H$ can be appropriately embedded into a random regular
%% %% graph $G$ on $O(n \log n)$ vertices.
%% The specific meaning of
%% ``appropriately embedded'' here is that
with the two additional properties that (i) every path used in
the embedding of $H$ has odd length, and (ii) there exists a perfect
matching in the subgraph of $G$ induced by the vertices \emph{not}
used in the embedding of $H$, then we obtain
average-case lower bounds for
the perfect matching formula $\pmf{G} \equiv \Card(G, \vec{1})$.
The lower bounds for
$\Card(G, \vec{t})$ for $t > 1$ can then be obtained by a reduction to
the $t=1$ case: after fixing the value of the edges in $\floor{t/2}$
cycle covers of $G$ to $1$, a restriction of $\Card(G, \vec{t})$ is
obtained which behaves like $\Card(G', \vec{1})$ for a somewhat
sparser random regular graph $G'$.

Let us elaborate a bit further on the properties required from the
topological minor of $H$ in $G$. As mentioned previously, our
embedding theorem can ensure that all paths are of odd length.
To ensure the second property, we in fact do not
embed $H$ directly into $G$ but rather into a suitably chosen vertex
induced subgraph $G[T]$ with the crucial property that for any set of
vertices $U \subseteq T$ of odd cardinality the induced subgraph $G[V \setminus U]$
has a perfect matching. As the embedding of $H$ will
consist of an odd number of vertices we then obtain property (ii) above.
Since we now want to apply \cref{thm:odd-minor-informal} not to $G$ but to
$G[T]$, we have to ensure that $G[T]$ satisfies all the conditions of that
theorem. We prove what we refer to as the \hyperref[lem:partition]{Partition Lemma},
which asserts that an induced subgraph $G[T]$ exists that satisfies both the perfect matching
property described above, as well as all conditions of
\cref{thm:odd-minor-informal}.  The proof of the Partition Lemma
relies primarily on the Lovász Local Lemma and spectral bounds to
obtain the desired properties.

For the proof of our embedding theorem (\cref{thm:odd-minor-informal}), we
extend an argument
due to Krivelevich and Nenadov~\cite{kn2018MinorsNoCut}
(see also \cite{krivelevich2018expanders})
for ordinary minors (rather than topological minors).
In order to obtain a minor embedding of $H$ in $G$, the idea there
is to embed the vertices one by one from $H$ in $G$  while maintaining an ``unused''
subgraph $G'$ of $G$ which is a slightly worse expander than $G$ is.
During this process it may happen that some vertex embedding cannot be
connected to a neighbor. If this happens, the embedding of that vertex is removed and
it needs to be embedded again.

In order to obtain topological embeddings, we need to adapt this
procedure.  Since we now want vertex-disjoint paths connecting the
embedded vertices, we would ideally like to embed each vertex of $H$
as a large star, and then embed the edges of $H$ as paths connecting
different leaves of these stars.  In order to make this work out,
rather than embedding the vertices as actual stars, we embed them as
``star-like'' subgraphs of $G$ (more precisely defined in
\cref{def:cross}) that consist of a central vertex connected to many
large vertex-disjoint connected subgraphs of $G$ and show
(\cref{lem:cross}) that we can always embed the vertices of $H$ as such
``star-like'' subraphs of $G$.

With this in place, obtaining control of the parities of the path
lengths used in the embedding (under the assumption on odd cycles in
\cref{thm:odd-minor-informal}) is relatively straightforward: almost
by definition, when embedding an edge of $H$ into a path of $G$, we
can route it via an odd cycle and can then choose which of the two
halves of the odd cycles to use, obtaining two possible embeddings
with different path length parity, and can choose the one with the
appropriate parity.

\subsection{Organization}

We give some preliminaries in \cref{sec:preliminaries}, formally
defining the used proof systems and encodings used, and recalling some
general background results.  In \cref{sec:avg-case} we provide
most of the proof of \cref{thm:pm avg lower bound} while deferring
the proofs of two key results, the aforementioned Partition Lemma and
our embedding theorem.  The proof of the Partition Lemma is given in
\cref{sec:partition}, and the proof of the embedding theorem can be
found in \cref{sec:odd-minor}.

In \cref{sec:worst-case} we recall the reduction of Buss et
al.~\cite{BGIP01LinearGaps} from Tseitin to perfect matching and show that it
yields lower bounds not only for Polynomial Calculus but also for
Sum-of-Squares and bounded depth Frege.

\section{Preliminaries}
\label{sec:preliminaries}

Natural logarithms (base $\mathrm{e}$) are denoted by $\ln$, whereas base
$2$ logarithms are denoted by $\log$.
For integers $n \ge 1$ we introduce the shorthand
$[n] = \set{1, 2, \ldots, n}$ and sometimes identify singletons
$\set{u}$ with the element $u$. For a set $U$ we denote the power set of
$U$ by $2^{U}$ and a transversal $A$ of a family of sets
$\mathcal{B} = \set{B_1, B_2, \ldots B_n}$ is a set such that there is
a bijective function $f : A \to \mathcal{B}$ satisfying that
$a \in f(a)$ for all elements $a \in A$.

\subsection{Proof Systems}
\label{sec:pcplx}

Let
$\calP = \set{p_1 = 0, \ldots, p_m = 0}$ be a system of polynomial
equations over the set of variables
$X = \set{x_1, \ldots, x_n, \bar{x}_1, \ldots, \bar{x}_n}$.
Each $p_i$ is called an axiom, and
throughout the paper we always assume $\calP$ includes all axioms
$x_i^2 - x_i$ and $\bar{x}_i^2 - \bar{x}_i$, ensuring that the variables
are boolean, as well as the axioms $1 - x_i - \bar{x}_i$, making sure that the
``bar'' variables are in fact the negation of the ``non-bar''
variables.

\paragraph{Sum-of-Squares (SoS)}
is a static semi-algebraic proof system.
An SoS proof of $f \ge 0$ from $\calP$ is a sequence of polynomials
$\pi = (t_1, \ldots, t_m; s_1, \ldots, s_a)$ such that
\begin{align}
  \sum_{i \in [m]} t_i p_i +
  \sum_{i \in [a]} s^2_i = f \eqperiod \label{eq:sos-def}
\end{align}
The \emph{degree} of a proof $\pi$ is
\begin{align}
  \Deg(\pi) =
  \max
  \set{
    \max_{i \in [m]}
    \deg(t_i) + \deg(p_i),
    \max_{i \in [a]}
    2 \deg(s_i)
  } \eqperiod
\end{align}
An \emph{SoS refutation of $\calP$} is an SoS proof of $-1 \ge 0$
from $\calP$, and the SoS degree to refute $\calP$ is the minimum
degree of any SoS refutation of $\calP$: if we let $\pi$ range over
all SoS refutations of $\calP$, we can write
$\RefuteDeg{SoS}{\calP} = \min_\pi \Deg(\pi)$.

\begin{definition}[Pseudoexpectation]
  A degree $d$ pseudo-expectation for $\calP$
  is a linear operator $\pedef$ on the
  space of real polynomials of degree at most $d$, such that
  \begin{enumerate}[label=$(\roman*)$]
  \item $\pe{1} = 1$,
  \item $\pe{tp} = 0$ for all polynomials $t$ and $p \in \calP$ with
    $\deg(t) + \deg(p) \le d$, and
  \item $\pe{s^2} \ge 0$ for all polynomials $s$ of degree $\deg(s) \le d/2$.
  \end{enumerate}
\end{definition}

It is easy to check that if there is a degree $d$ pseudo-expectation
for $\calP$, then there is no SoS refutation of $\calP$ of degree at
most $d$: if $\pedef$ is applied to both sides of \eqref{eq:sos-def},
where $f = -1$, then the right side is equal to $-1$ while the left is
greater than or equal to $0$.

The size of an
SoS refutation $\pi$, $\Size(\pi)$,
is the sum of the number of monomials in each
polynomial in $\pi$ and the size of refuting $\calP$ is the minimum
size over all refutations
$\RefuteSize{SoS}{\calP} = \min_{\pi} \Size(\pi)$.

\paragraph{Polynomial Calculus} is a dynamic proof system
operating on polynomial equations over a field $\mathbb{F}$. Let
$\calP$ be over $\mathbb{F}$.
Polynomial Calculus over $\mathbb{F}$ (PC$_{\mathbb{F}}$)
consists of the derivation rules
\begin{itemize}
\item linear combination
  \AxiomC{$p = 0$}\AxiomC{$q = 0$}\BinaryInfC{$\alpha p + \beta q = 0$}\DisplayProof,
  where $p, q \in \mathbb{F}[X]$ and $\alpha, \beta \in \mathbb{F}$, and
\item multiplication \AxiomC{$p = 0$}\UnaryInfC{$xp = 0$}\DisplayProof,
  where $p \in \mathbb{F}[X]$ and $x \in X$.
\end{itemize}
A PC refutation of $\calP$ is a sequence of polynomials
$\pi = t_1, \ldots, t_\ell$
such that $t_\ell = 1$ and each polynomial $t_i$ is
either in $\calP$ or can be derived by one of the derivation rules
from earlier polynomials. The degree of a refutation is the maximum
degree appearing in the sequence $\Deg(\pi) = \max_{i \in [\ell]}
\Deg(t_i)$ and the PC$_{\mathbb{F}}$ degree of refuting $\calP$ is the
minimum degree required of any refutation
$\RefuteDeg{PC$_{\mathbb{F}}$}{\calP} = \min_{\pi} \Deg(\pi)$.
Similarly, the size of a refutation $\pi$ is the sum of the number of
monomials in each line of $\pi$ and the
PC$_{\mathbb{F}}$ size of refuting $\calP$ is the
minimum size required of any refutation
$\RefuteSize{PC$_{\mathbb{F}}$}{\calP} = \min_{\pi} \Size(\pi)$.

\paragraph{Frege System}
Let us describe a Frege system due to Shoenfield, as presented in
\cite{UF96SimplifiedLowerBounds}.
As Frege systems over the basis
$\lor$, $\land$ and $\lnot$
can polynomially simulate each other
\cite{cook79efficiency},
the details of the system are not essential and hold for any Frege
system over the mentioned basis.

Schoenfield's Frege system works over the basis
$\lor$ and $\lnot$. We treat
the conjunction $A \land B$ as an abbreviation for
the formula $\lnot(\lnot A \lor \lnot B)$ and let 0, 1 denote
``false'' and ``true'' respectively.
If $A$ is a formula over variables $p_1, \ldots, p_m$, and
$\sigma$ maps the variables $p_1, \ldots, p_m$ to formulas
$B_1, \ldots, B_m$, then $\sigma(A)$ is the formula obtained from $A$ by
replacing the variable $p_i$ with $B_i = \sigma(p_i)$ for all
$i \in[m]$.

% rule
A \emph{rule} is a sequence of formulas written as
$A_1, \ldots, A_{k} \vdash A_0$.
% sound (dropped axiom).
If every truth assignment
satisfying all of $A_1, \ldots, A_{k}$ also satisfies $A_0$, then the
rule is \emph{sound}.
% application
A formula $C_0$
is inferred from $C_1, \ldots, C_k$ by the rule
$A_1, \ldots, A_k \vdash A_0$ if there is a function
$\sigma$ mapping the variables
$p_1, \ldots, p_m$, over which $A_0, \ldots, A_k$ are defined, to
formulas $B_1, \ldots, B_m$ such that for all
$i \in \set{0, \ldots, k}$ it holds that
$C_i = f(A_i)$.

The Frege system \frege that we consider consists of the following
rules:
\begin{align*}
  &\vdash p \lor \lnot p &&\text{Excluded Middle,}\\
  p &\vdash q \lor p &&\text{Expansion rule,}\\
  p \lor p &\vdash p &&\text{Contraction rule,}\\
  p \lor (q \lor r) &\vdash (p \lor q) \lor r
  && \text{Associative rule,}\\
  p \lor q, \lnot p \lor r &\vdash q \lor r
  && \text{Cut rule.}
\end{align*}

An \emph{\frege-refutation} of an unsatisfiable formula
$A = C_1 \land \ldots \land C_m$ is a sequence of formulas
$F_1, F_2, \ldots, F_\ell$ such that
$F_\ell = 0$ and
every formula $F_i$
is either one of $C_1, \ldots, C_m$ or inferred from
formulas
$F_{j_1}, \ldots, F_{j_k}$ earlier
in the sequence by a rule in \frege.
As \frege is sound and complete a formula $A$ has a
refutation if and only if it is unsatisfiable.

The size of a formula is the number of connectives in the formula and
the size of a refutation~$\pi$, denoted by $\Size(\pi)$, is the sum of
the sizes of all formulas in the refutation. The depth of $\pi$
is the maximum depth of any formula $F \in \pi$. We denote by
\frege[d] the proof system \frege restricted to formulas of
depth at most $d$.

\subsection{Propositional Formulas}
\label{sec:formulas}

As we are only interested in constant degree graphs all our axioms
are of constant size. Hence the precise encoding of the axioms is not
significant as we can change the encoding in constant size/degree.

As the encoding is not essential, we view a propositional formula $\calF$
over the Boolean variables $x_1, \ldots, x_n$
as a family of
functions $\calF = \set{f_1, \ldots, f_m}$ where each
$f_i: \set{0,1}^n \to \set{\text{True},\text{False}}$
is a function that
depends on a constant
number of variables.
The formula $\calF$ is satisfied by an assignment
$\alpha \in \set{0,1}^n$ if under $\alpha$ all
functions evaluate to True: $f_i(\alpha) = \text{True}$ for all
$i \in [m]$.

For a map
$\rho:
\set{x_1, \ldots, x_n} \to
\set{0,1, x_1, \ldots, x_n, \bar{x}_1, \ldots, \bar{x}_n}$
and a function
$f: \set{0,1}^n \to \set{\text{True}, \text{False}}$,
denote by $\restrict{f}{\rho}$ the function defined by
$\restrict{f}{\rho}(x_1, \ldots, x_n) = f(\rho(x_1), \ldots, \rho(x_n))$.
We extend
this notation to formulas in the obvious way, i.e., $\restrict{\calF}{\rho} =
\set{\restrict{f_1}{\rho},
  \restrict{f_2}{\rho}, \ldots,
  \restrict{f_m}{\rho}
}$.

Two formulas $\calF$ and $\calF'$ are equivalent, denoted by
$\calF \equiv \calF'$ if the formulas are element-wise equivalent,
disregarding functions that are constant True.
We say that a formula $\calF'$ is an \emph{affine restriction of $\calF$} if there is a map
$\rho:
\set{x_1, \ldots, x_n} \to
\set{0,1, x_1, \ldots, x_n, \bar{x}_1, \ldots, \bar{x}_n}$
such that
$\calF' \equiv \restrict{\calF}{\rho}$.
The following lemma states that a formula $\calF$ is at least as hard
as any of its affine restrictions.

\begin{lemma}\label{lem:subformula}
  Let $\calF, \calF'$ be formulas such that $\calF' \subformula \calF$
  and each axiom of~$\calF$ depends on a constant number of
  variables. Then,
  \begin{enumerate}[label=$(\roman*)$]
  \item for any field $\mathbb{F}$ it holds that
    $\RefuteDeg{PC$_{\mathbb{F}}$}{\calF} \in
    \Omega\big(\RefuteDeg{PC$_{\mathbb{F}}$}{\calF'}\big)$,
  \item $\RefuteDeg{SoS}{\calF} \in
    \Omega\big(\RefuteDeg{SoS}{\calF'}\big)$, and
  \item
    for all $d \ge 2$ it holds that
    $\RefuteSize{\frege[d]}{\calF} \in
    \Omega\big(\RefuteSize{\frege[d+1]}{\calF'}\big)$.
  \end{enumerate}
\end{lemma}

\begin{proof}
  Suppose we have a refutation $\pi$ of $\calF$ in one of the
  mentioned proof systems. We want to show that if we hit the proof
  with the restriction $\rho$ such that $\restrict{\calF}{\rho} \equiv \calF'$
  then we obtain a proof $\pi' = \restrict{\pi}{\rho}$ of $\calF'$.

  First we need to ensure that we can derive all the axioms of
  $\calF'$. These may be encoded in a different manner, but as
  these proof systems are implicationally complete, and each axiom
  only depends on a constant number of variables, this can be done in
  constant degree (constant size).

  This shows that the SoS degree of the resulting refutation is at
  most a constant factor larger. For Polynomial Calculus and Frege the
  statement is readily verified by an inductive argument over the
  proof.
\end{proof}

For concreteness let us also define the encoding of the formulas that we
are interested in.

\paragraph{Perfect Matching and $\Card(G, \vec{b})$}
The Perfect Matching formula $\pmf{G}$ encodes the claim that the
graph $G$ contains a perfect matching. For every edge $e \in E(G)$
introduce a boolean variable $x_e \in \set{0, 1}$ and add for every
vertex $v \in V(G)$ an axiom claiming that precisely one incident edge
is set to true. As a polynomial over $\R$, we encode this claim as
\begin{align}
  \pmax{v} = \sum_{e \ni v} x_e - 1 \eqcomma
\end{align}
which is satisfied under an assignment $\alpha$
if $\pmax{v}(\alpha) = 0$. Over other fields we encode this as a sum over indicator
polynomials (see example for Tseitin below).  For the Frege proof system we encode the
vertex axiom as the propositional formula
\begin{align}
  \pmax{v} =
  \bigvee_{e \ni v} x_e
  \wedge
  \bigwedge_{
    \substack{
      e, e' \ni v\\
      e \neq e'}
  }
  \bar{x}_e \vee \bar{x}_{e'} \eqperiod
\end{align}

The formula $\Card(G, \vec{b})$ is encoded in a similar fashion: in
the polynomial encoding replace the 1 with $b_v$, whereas in the
propositional encoding we let the latter $\wedge$ range over edge-tuples of
size $b_v + 1$.

\paragraph{Tseitin Formula}
The Tseitin formula $\tsf{G}$ claims that the edges of the graph $G$
can be labeled by $0, 1$ such that the number of 1-labeled edges incident
to any vertex is odd.
For every edge $e \in E(G)$ introduce a boolean
variable $y_e \in \set{0, 1}$, denote the set of variables corresponding to edges
incident to $v$ by $Y_v = \set{y_e \mid v \in e}$ and let
$A_v \subseteq \set{0, 1}^{Y_v}$ contain all assignments to the
variables $Y_v$ that set an odd number of variables to 1.
We encode the claim that an odd number of edges incident to
$v \in V(G)$ are set to 1 as the polynomial
\begin{align}
  \tsax{v} = \sum_{\alpha \in A_v} \ind{Y_v = \alpha} - 1 \eqcomma
\end{align}
where
$
  \ind{Y_v = \alpha} =
\prod_{
  \substack{
    y \in Y_v\\
    \alpha(y) = 1}
} y
\prod_{
  \substack{
    y \in Y_v\\
    \alpha(y) = 0}
} \bar{y}
$
is the indicator polynomial that is 1 iff the variables in $Y_v$ are
set according to
$\alpha$. As before, we also add the boolean axioms to ensure that the
variables take values in $\set{0,1}$.

For the Frege system we encode the claim that an odd number of
edges incident to $v \in V(G)$ is set to $1$ as the propositional
formula
\begin{align}
  \tsax{v} = \bigvee_{\alpha \in A_v} \ind{Y_v = \alpha} \eqcomma
\end{align}
where the indicator is now encoded as the formula
$
  \ind{Y_v = \alpha} =
  \bigwedge_{
    \substack{
      y \in Y_v\\
      \alpha(y) = 1}
  } y
  \,\wedge\,
  \bigwedge_{
    \substack{
      y \in Y_v\\
      \alpha(y) = 0}
  } \bar{y}$.
%% If the numer of vertices is odd,
%% it is easy to check that $\tsf{G}$ is
%% unsatisfiable: just
%% double count the number of edges set to 1.

\subsection{Graph Theory}
\label{sec:prelim-graph}

% graphs
This paper only considers simple, undirected graphs:
all graphs have no self-loops nor multiple edges.
For a graph $G = (V, E)$
the neighborhood of a vertex $u \in V$ is
$N(u) = \setdescr{v \in V}{\set{u, v} \in E}$,
the neighborhood of a set of vertices $U \subseteq V$ is
$N(U) = \bigcup_{u \in U} N(u)$ and for sets
$U, W \subseteq V(G)$ the neighborhood of
$U$ in $W$ is $N(U, W) = N(U) \cap W$.
We denote by $\deg(v) = |N(v)|$ the degree of a vertex $v \in V$, by
$\Delta(G)$ the maximum degree, $\delta(G)$ the minimum degree
and by $d(G)$ the average degree of $G$.
The edges between two vertex sets $U, W \subseteq V$ are denoted
by
$
  E(U, W) = \setdescr{\set{u, w} \in E}{u \in U, w \in W}% \eqperiod
$.
% induced subgraph
For a set $U \subseteq V$, we denote by $G[U] = (U, E(U, U))$ the
\emph{induced subgraph} of $U$ in $G$. For a
set $T \subseteq V$ we also use $G \setminus T$ as a shorthand for
the induced subgraph $G[V \setminus T]$.
%
% paths
For a path $p$ in $G$ we denote by $|p|$ the number of edges and by
$V(p) \subseteq V(G)$ the set of vertices of $p$.  For two vertices
vertices $u, v \in V(p)$, we let $p[u, v]$ denote the subpath of $p$
between (and including) the vertices $u$ and $v$.
% diameter
The distance between two vertices $u, v \in V$ is the length of the
shortest path from $u$ to $v$ and the distance between two sets
$U, W \subset V$ is the minimum distance between any pair of vertices
$u \in U$ and $w \in W$.  Let $\diam{G}$ denote the diameter of $G$,
that is, the maximum distance between any two vertices in $G$.
% Ball def
For a vertex set
$U \subseteq V$, and an integer $r \in \N$, let
$B^G_r(U) \subseteq V(G)$
be the \emph{ball around $U$ of radius $r$ in $G$}:
$B^G_r(U)$ contains all vertices
$v \in V$ that are at distance at
most $r$ from $U$.

% expansion
A graph $G$ on $n$ vertices
is an \emph{$\alpha$-expander}
(has \emph{vertex expansion} $\alpha$)
if for all sets $U \subseteq V(G)$ of size
$|U| \le n/2$
it holds that
$|N(U, V \setminus U)| \ge \alpha |U|$.
%
% d regular random graph
We denote the
\emph{uniform distribution
over $d$-regular graphs on $n$ vertices}
by $\calG(n,d)$ and tacitly assume throughout this paper
that $n d$ is even.
%
% topological minor
A graph $G$ contains $H$ as a \emph{topological minor} if there is an
injective map $\sigma: V(H) \to V(G)$ and for every
$\set{u,v} \in E(H)$ there is a path $p_{u v} \subseteq G$ from
$\sigma(u)$ to $\sigma(v)$ that is pairwise vertex-disjoint from all
other paths except in the endpoints. The paths $p_{u v}$ are the
\emph{edge embeddings} of the minor.
% there is a problem with isolated vertices. According to the def
% these may appear on paths.

Let us record the well-known fact that vertex expanders have small
diameter.

\begin{lemma}[\cite{krivelevich2018expanders}]\label{lem:exp-diam}
  Let $G$ be an $\alpha$-expander on $n$ vertices. Then the diameter
  of $G$ is upper bounded by
  $\big\lceil\frac{2(\log{n}-1)}{\log(1+\alpha)}\big\rceil + 1
  = O_\alpha(\log{n})$.
\end{lemma}

As this constant will show up in a few places, let
$\cdiam = \frac{2}{\log(1+\alpha)} + 3$ and hence
$\diam{G} \le \cdiam \cdot \log{n}$, if $G$ is an
$\alpha$-expander.
% If $\alpha = 0$, then we simply set $\cdiam[0] = \infty$.

The following lemma states that even if a small set of vertices is
removed from a vertex expander, large sets still have many vertices at
small distance.

\begin{lemma}\label{lem:exp-ball-remove}
  Let $G$ be an $\alpha$-expander on $n$ vertices. Then for all
  $r \ge 0$ and all disjoint $S, T \subseteq V(G)$ satisfying
  $|T| \ge \frac{2}{\alpha} |S|$ it holds that
  $|B_r^{G \setminus S}(T)| \ge \min \set{n/2, (1 + \alpha/2)^r |T|}$. 
\end{lemma}

\begin{proof}
  Using expansion and $|S| \le \frac{\alpha}{2} |T| \le
  \frac{\alpha}{2} |B_{r}^{G \setminus S}(T)|$ we have that for all $r
  \ge 0$
  \begin{align*}
    |B_{r+1}^{G \setminus S}(T)|
    &\ge (1+\alpha)|B_{r}^{G \setminus S}(T)| - |S|
    \ge (1+\alpha/2)|B_{r}^{G \setminus S}(T)| \eqcomma
  \end{align*}
  unless $B_r^{G \setminus S}(T)$ is already as large as $n/2$.
\end{proof}

A simple consequence of this is that two large sets are
connected by short paths even after the removal of a small set of vertices.

\begin{corollary}\label{cor:exp-diam-remove}
  Let $G$ be an $\alpha$-expander on $n$ vertices. Then for all sets
  $S, T, U \subseteq V(G)$ satisfying that $T, U \neq \emptyset$, that
  $S \cap (T \cup U) = \emptyset$, and
  $|T|, |U| \ge \frac{2}{\alpha} |S|$ it holds that in $G \setminus S$
  the distance between $T$ and $U$ is at most
  $\cdiam[\alpha/2]\log n$.
\end{corollary}

\begin{proof}
  Apply \cref{lem:exp-ball-remove} to $S, T$ and
  $r = \ceil{\frac{\log n}{\log(1 + \alpha/2)}}$ to conclude that at
  distance $r$ from $T$ there are at least $n/2$ vertices in the graph
  $G \setminus S$. Applying the same argument to $U$ and $S$, we see
  that also from $U$ there are at least $n/2$ vertices reachable by
  length $r$ paths in $G \setminus S$. But this implies that there is
  a path of length at most $2r + 1 \le \cdiam[\alpha/2]\log n$ between
  $T$ and $U$.
\end{proof}

\subsection{Probabilistic Bounds}

We use the following version of the multiplicative Chernoff bound.

\begin{theorem}[Chernoff]\label{thm:chernoff}
  Suppose $X_1, \ldots, X_n$ are independent random variables taking
  values in $\set{0, 1}$. Let $X$ denote their sum and let
  $\mu = \E[X]$.
  Then, for every $0 \le \delta \le 1$ we have
  \begin{align*}
    \prob{\abs{X - \mu} \ge \delta\mu} \le 2 \exp(-\delta^2\mu/3) \eqperiod
  \end{align*}
\end{theorem}

We also need a similar bound for Poisson random variables.

\begin{theorem}[\cite{mu2005}, Theorem 5.4]
  \label{thm:poi-tail}
  Let $X$ be a Poisson random variable with parameter $\mu$. If
  $x > \mu$, then
  \begin{align*}
    \prob{X \ge x}
    \le e^{-\mu} \left(\frac{e \mu}{x}\right)^x \eqperiod
  \end{align*}
\end{theorem}

Finally we also need the following form of the Lovász local lemma.

\begin{lemma}[Lovász local lemma; \cite{alonspencer}, Lemma
  5.1.1]\label{lem:lll}
  Let $A_1, A_2, \ldots, A_n$ be events in an arbitrary probability
  spacce. A directed graph $D = (V, E)$ on the set of vertices
  $V=\set{1, 2, \ldots n}$ is called a dependency digraph for the
  events $A_1, \ldots, A_n$ if for each $i$, $1 \le i \le n$, the
  event $A_i$ is mutually independent of all the events
  $\set{A_j \mid (i, j) \not\in E}$. Suppose that $D = (V, E)$ is a
  dependency digraph for the above events and suppose there are real
  numbers $x_1, \ldots x_n$ such that $0 \le x_i < 1$ and
  $\prob{A_i} \le x_i \prod_{(i, j) \in E}(1-x_j)$
  for all $1 \le i \le n$. Then
  $\prob{\wedge_{i=1}^n \bar{A_i}} \ge \prod_{i = 1}^n(1-x_i)$.
\end{lemma}

\section{Lower Bounds on Average}
\label{sec:avg-case}

In this section we establish our main result \cref{thm:pm avg lower
  bound} giving average-case lower bounds in PC, SoS and bounded depth
Frege for the $\Card(G, \vec{t})$ formulas.

\subsection{Lower Bounds for Perfect Matching}

Recall that we aim to prove that any sparse graph $H$ (in particular a
graph where $\pmf{H}$ is hard to refute) can be topologically embedded
into a random graph such that all paths in the embedding have odd
length.  In order to do this, we need to assume that the graph is far
from bipartite (since otherwise $H$ would need to be bipartite as
well, and $\pmf{H}$ is easy for bipartite graphs).  Furthermore our
embedding theorem relies on all large induced subgraphs of $G$ having
sufficiently large maximum degree.  The two following definitions
capture that both properties hold for all large induced subgraphs of
$G$.

\begin{definition}\label{def:high-degree-rich}
  A graph $G$ on $n$ vertices is \emph{\degrich{\kappa}{d}} if for all
  $U \subseteq V(G)$ of size $|U| \ge \kappa n$ it holds that the
  maximum degree of the induced subgraph $G[U]$ is
  $\Delta(G[U]) \ge d$.
\end{definition}

\begin{definition}\label{def:odd-cycle-rich}
  A graph $G$ on $n$ vertices is \emph{\nonbip{\kappa}{\alpha, \ell}}
  if for all $U \subseteq V(G)$ of size $|U| \ge \kappa n$ and such
  that $G[U]$ is an $\alpha$-expander it holds that the induced
  subgraph $G[U]$ contains an odd cycle $C$ of length
  $\ell \le |C| \le 3 \cdiam[\alpha/2] \log n$.
\end{definition}

Note that in the latter definition, assuming that $G[U]$ is an
$\alpha$-expander, the diameter of $G[U]$ is at most
$\cdiam[\alpha] \log(n) \le \cdiam[\alpha/2]\log(n)$ which means that,
unless $G[U]$ is bipartite, it certainly has short odd cycles of
length at most $1+2\cdiam[\alpha/2] \log n$. But a priori these may
all be shorter than~$\ell$. The definition asks for short odd cycles
of length at least $\ell$, at the cost of a slightly worse upper bound
on the cycle length.

Both properties are clearly monotone in $\kappa$: if the properties
hold for some $\kappa_0 > 0$, then they also hold for all
$\kappa \ge \kappa_0$.
With these definitions at hand we can state our embedding theorem.

\begin{restatable}[Embedding Theorem]{theorem}{EmbeddingTheorem}
  \label{thm:odd-minor}
  For $\alpha > 0$
  there are $\epsilon, n_0 > 0$ such
  that the following holds.
  Let $G$ be an
  $\alpha$-expander on $n > n_0$ vertices,
  let $k \ge 6$,
  % $k = k(\alpha, G) \ge 6$,
  and let $H$ be a graph on at most
  $\epsilon n / k \log n$
  vertices and edges.
  If $G$ is \degrich{1-4/k}{550\Delta(H)/\alpha^2},
  then $G$ contains~$H$ as a
  topological minor.
  Furthermore, if $G$ is also
  \nonbip{1 - 2/k}{\beta, 1 + 2/\beta},
  for $\beta = \frac{\alpha}{3(1+\alpha)}$,
  then one can choose the parities
  of the lengths of all the edge embeddings in the minor.
\end{restatable}

Let us highlight that $k$ may depend on the graph $G$. We have made
no attempt to optimize the constants.
The proof of the embedding theorem can be found in \cref{sec:odd-minor}.

As mentioned before we need to ensure that once we obtain an
embedding of the worst-case graph $H$ in $G$, that there is a
matching in the graph $G$ with the embedding of $H$ removed. To ensure
this we will in fact not embed $H$ directly in $G$
but rather in a subgraph of $G$: first we identify a set of vertices
$T \subseteq V(G)$ such that no matter what set $U \subseteq T$ of odd
cardinality is removed from $G$, the graph $G \setminus U$ still
contains a perfect matching.
We then proceed to show that the graph $G[T]$ satisfies all the
properties required in order to embed $H$ into it.
The following lemma captures these properties\footnote{
  For clarity of exposition
  we say that an event holds for \emph{odd} $n$
  asymptotically almost surely as $n \rightarrow \infty$
  if $n = 2n' + 1$ for some non-negative integer $n'$ and the event
  holds asymptotically almost surely as $n' \rightarrow \infty$.}.

\begin{restatable}[Partition Lemma]{lemma}{PartitionLemma}
  \label{lem:partition} \label{LEM:PARTITION}
  There is a
  $d_0$ such that for all
  $d > d_0$ there is an
  $n_0$ such
  that the following holds.
  Let $n > n_0$ be odd and $G \sim \mathcal{G}(n, d)$.
  Then,
  asymptotically almost surely,
  there is a set
  $T \subseteq V(G)$
  of size $|T| \ge n/8$
  such that
  $G[T]$ is
  a
  $1/3$-expander,
  \nonbip{1/2}{1/12, 25},
  \degrich{1/3}{d/32}
  and for any set $U \subseteq T$ of odd
  cardinality it holds that
  $G \setminus U$ has a perfect matching.
\end{restatable}

The partition lemma is proved in \cref{sec:partition}.  The constants
in \cref{lem:partition} are rather arbitrarily chosen and their
precise values are not significant -- the interested reader can find
the precise dependencies between them in the proof.  With
\cref{thm:odd-minor,lem:partition} at hand, we can now easily state
and prove our lower bounds for the perfect matching formula (i.e., the
special case $t=1$ of \cref{thm:pm avg lower bound}).

\begin{theorem}\label{thm:pm-avg}
\label{THM:PM-AVG} %Paul
  There is a $d_0$
  and an $\eps > 0$ such that for all
  $d > d_0$ the following holds.
  For $n$ and
  $n' \le \frac{\eps n}{\log n}$
  both odd, let
  $G \sim \mathcal{G}(n,d)$ and
  $H$ be any
  graph on $n'$ vertices of degree
  $\Delta(H) \le 5$.
  Then, asymptotically almost surely,
  $\pmf{H} \subformula \pmf{G}$.
\end{theorem}

Using the graphs from \cref{sec:worst-case} (i.e., the graphs from
\cref{thm:worst-pc,thm:worst-sos,thm:worst-frege}) as our choice of $H$
and combining \cref{thm:pm-avg} with
\cref{lem:subformula} finishes the proof of
\cref{thm:pm avg lower bound} for the perfect matching formula.

% \begin{proof}[Proof of \cref{thm:pm-avg}]
% \end{proof}%Paul
\begin{proof}[Proof of Theorem~\ref{thm:pm-avg}]%Paul
  Let $G \sim \mathcal{G}(n, d)$ as in the statement.
  Apply \cref{lem:partition} to $G$
  to obtain a set $T$
  with the mentioned properties.
  In order to apply \cref{thm:odd-minor} to
  $G[T]$ and the graph $H$
  to obtain a
  topological minor
  $\embed{H} \subseteq G[T]$,
  where all edge embeddings in $\embed{H}$
  are of odd length,
  we need to check that (for our choice $\alpha = 1/3, k=6$)
  \begin{enumerate}[label=$(\roman*)$]
  \item $G[T]$ is a $1/3$-expander,\label{item:expansion-check}
  \item $G[T]$ is \degrich{1/3}{550 \cdot 5 \cdot 9},\label{item:degrich-check}
  \item $G[T]$ is \nonbip{2/3}{1/12, 1 + 2 \cdot 12}, and\label{item:nonbip-check}
  \item $H$ is a graph on at most $\epsilon n /6 \log n$
    vertices and edges, for some $\epsilon > 0$. \label{item:H-check}
  \end{enumerate}
  From the guarantees of \cref{lem:partition} we see that
  \ref{item:expansion-check} is satisfied,
  that for $d$ large \ref{item:degrich-check} holds
  and also that
  \ref{item:nonbip-check} holds
  as odd-cycle-robustness is monotone in the first argument.
  Lastly, \ref{item:H-check} holds if we let $\eps = \epsilon/6$.

  With the topological minor $\embed{H}$ of $H$ in $G$ at hand, we proceed
  to construct a restriction $\rho$ to argue that $\pmf{H} \subformula \pmf{G}$.
  As all edge embeddings in $\embed{H}$ are of odd length
  and the number of vertices in $H$ is odd,
  we see that
  $|V(\embed{H})|$ is odd.
  Hence \cref{lem:partition} guarantees
  that there exists a perfect
  matching $M$ in the graph
  $G' = G \setminus V(\embed{H})$.
  The restriction~$\rho$ sets
  all variables outside of $\embed{H}$ to $0$ or $1$ depending on whether
  the edge $e \in M$.

  We still need to specify how $\rho$ maps the variables in
  $\embed{H}$. For every edge embedding
  $p_{u v}$ of~$\embed{H}$,
  choose an arbitrary edge $e_{u v} \in p_{u v}$ and map the
  edge variables $x_e$, for $e \in p_{u v}$,
  alternatingly along $p_{u v}$ to either
  $x_{e_{u v}}$ or $\bar{x}_{e_{u v}}$
  such that the first and last edge of
  $p_{u v}$ are mapped to $x_{e_{u v}}$
  (where we use that $|p_{u v}|$ is odd).
  By inspection we see that
  $\pmf{H} \subformula \pmf{G}$ as claimed.
\end{proof}

\subsection{Lower Bounds for $\Card(G, \vec{t})$}
\label{sec:avg-general}

In the following we prove the average-case lower bounds
on the $\Card(G, \vec{t})$ formulas for $G \sim \calG(n,d)$.
We consider the special case when
$n$ and $t \le d$ are odd and thus $d$ is even.
Without loss of generality, assume that $t \le d/2$: otherwise ``flip'' the
roles of $0$ and $1$.

The idea is to split
the edge set of the graph $G$
into $\floor{t/2}$
$2$-regular graphs
$G_1, \ldots, G_{\floor{t/2}}$
and one $d_0$-regular graph $G_0$, where $d_0=d-2\floor{t/2}$.
Then we want to set all variables that correspond to an edge in any of
the 2-regular graphs $G_1, \ldots, G_{\floor{t/2}}$ to 1 so that we
are left with the perfect matching formula $\pmf{G_0}$, on which we
will embed the worst-case instance of \cref{sec:worst-case}.

In order to be able to apply \cref{thm:pm-avg} to $\pmf{G_0}$, we need
to argue that $G_0$ is a random $d_0$-regular graph. Also, we need to
show that it is in fact possible to decompose a random $d$-regular graph
into $\floor{t/2}$ $2$-regular graphs plus a $d_0$-regular graph.
For this, we
use the notion of \emph{contiguity}.
Intuitively, two
sequences of probability measures are contiguous, if
all properties that hold with high probability in one
also hold with high probability in the other measure.

\newcommand{\pmP}{P_n}
\newcommand{\pmQ}{Q_n}
\newcommand{\sequ}[1]{(#1)_{1}^{\infty}}
\begin{definition}
  Let $\sequ{\pmP}$ and $\sequ{\pmQ}$ be two sequences of probability
  measures, such that for each $n$, $\pmP$ and $\pmQ$ both are defined
  on the same measurable space $(\Omega_n, \mathcal{F}_n)$.
  The two sequences are \emph{contiguous}
  if for every sequence of sets $\sequ{A_n}$, where
  $A_n \in \mathcal{F}_n$, it holds that
  \begin{align*}
    \lim_{n \rightarrow \infty} \pmP(A_n) = 0 \Leftrightarrow \lim_{n
    \rightarrow \infty} \pmQ(A_n) = 0 \eqperiod
  \end{align*}
  We denote contiguity of two sequences by $\pmP \approx \pmQ$.
\end{definition}

\newcommand{\graphplus}{\oplus}
For two random graphs $\calG_n$ and $\calH_n$ on the same set of $n$ vertices, we denote by
$\calG_n \graphplus \calH_n$ the union of two independent samples conditioned
on the result being simple. If $\calG_n = \calG(n, d)$ and
$\calH_n = \calG(n, d')$ are uniform distributions over
random regular graphs we can think of
this as a proccess where
we first sample $G \sim \calG_n$ and then repeatedly sample
$H \sim \calH_n$ until the union of~$G$ and~$H$ is simple.

\begin{theorem}[Corollary 9.44, \cite{JLR00}]\label{thm:contiguous}
  For all constants $d \ge 3$, $m \ge 1$ and $d_1, \ldots, d_m \ge 1$
  satisfying $d = \sum_{i=1}^m d_i$
  it holds that
  $$
  \calG(n, d_1) \graphplus \cdots \graphplus \calG(n, d_m)
  \approx
  \calG(n, d) \eqperiod
  $$
\end{theorem}

In other words, if we can show that e.g.~SoS
requires linear degree for a formula over
$
G
\sim
\calG(n, d_0) \graphplus
\bigoplus_{i \in \floor{t/2}} \calG(n, 2)
$
with high probability, then this
also holds for the same formula over graphs
$G \sim \calG(n, d)$. Implementing our idea in the former
probability distribution is straightforward
and we have the following theorem.

\begin{theorem}\label{thm:avg-card}\label{THM:AVG-CARD}
  There is a $d_0$ and an $\eps > 0$ such that for all
  $d \ge d_0$ the following holds.
  Let $n, n' \le \frac{\eps n}{\log n}$ and
  $t \in [d]$ all be odd, let
  $G \sim \mathcal{G}(n,d)$ and $H$ be a
  graph on $n'$ vertices of degree
  $\Delta(H) \le 5$.
  Then, asymptotically almost surely,
  $\pmf{H} \subformula \Card(G, \vec{t})$.
\end{theorem}

Analogously to how \cref{thm:pm-avg} implied the $t=1$ case of
\cref{thm:pm avg lower bound}, this theorem implies the general case
of \cref{thm:pm avg lower bound}.

% \begin{proof}[Proof of \cref{thm:avg-card}]
\begin{proof}[Proof of Theorem~\ref{thm:avg-card}] %Uri
  As $n$ is odd
  $d$ must be even. Note that we may assume that $t \le d/2$: if $ t >
  d/2$, let us flip the role of 1 and 0 in the formula to obtain
  $\Card(G, \overrightarrow{d-t})$.
  Let $d_0 = d - 2\floor{t/2} \ge d/2$ and sample
  \begin{align}
    G' = G_0 \cup \bigcup_{1 \le i \le \floor{t/2}} G_i
    \,\sim\,
    \calG(n, d_0) \graphplus \bigoplus_{1 \le i \le \floor{t/2}}\calG(n, 2)
    \eqperiod
  \end{align}
  By \cref{thm:contiguous}, if we show the statement for $G'$,
  then it also holds for $G \sim \calG(n, d)$.

  Set all variables in $G_1, \ldots G_{\floor{t/2}}$ to 1.
  When $\Card(G, \vec{t})$ is hit with
  this restriction we are left with the formula $\pmf{G_0}$. As $G_0$
  is distributed according to $\calG(n, d_0)$, we may apply
  \cref{thm:pm-avg} to conclude that
  $\pmf{H} \subformula \Card(G, \vec{t})$.
\end{proof}

\section{Proof of the Partition Lemma}
\label{sec:partition}

In this section we prove \cref{lem:partition}, restated here for convenience.

\PartitionLemma*

We proceed as follows. First, we partition $V(G) = S \disjointunion T$
into two sets such that every vertex $v \in V(G)$ has a good fraction of its
neighbors in $S$.

\begin{definition}
  A \emph{$(c, \eps)$-degree-balanced cut} of a graph $G$ is a partition $S
  \disjointunion T = V(G)$ of the $n$ vertices of $G$ such that:
  \begin{enumerate}[label=$(\roman*)$]
  \item $\big||S|-cn\big| \le \eps n$
  \item for every vertex $u \in V$, the fraction of $u$'s neighbors
    that are in $S$ is at least $c-\eps$ and at most $c+\eps$.
  \end{enumerate}
\end{definition}

It turns out that in random regular graphs any
$(c, \eps)$-degree-balanced cut possesses the properties needed in the
Partition Lemma, as summarized in the following lemma.

\begin{lemma}\label{lem:partition-properties}
  For all constants
  $c, \eps, d > 0$
  satisfying $c > 1/2 + \eps$ and
  $d \ge
  \max \{(c - 1/2 - \eps)^{-2}, 4\cdot\eps^{-2}\}$ the
  following holds.
  Let $n$ be odd and $G \sim \calG(n, d)$.
  Then, asymptotically almost surely as $n \rightarrow \infty$,
  for any $(c, \eps)$-degree-balanced cut
  $(S,T)$ of $G$ it
  holds that
  \begin{enumerate}[label=$(\roman*)$]
  \item the graph $G$ is
    \DEGRICH{\kappa}
    {
      d \big(
      \kappa -
      2\sqrt{\frac{1 - \kappa}{\kappa d}}
      \big)
    }
    for all constants $\kappa \in [0,1]$,
  \item the graph $G$ is \nonbip{6/\sqrt{d}}{\beta, \ell},
    for any constants $\beta$ and $\ell$,

  \item the graph $G[T]$ is an $\alpha$-expander, where
    $\alpha = \frac{1 - c - 2\eps}{2(1-c-\eps)}$, and

  \item the graph $G \setminus U$ has a perfect matching for any
    $U \subseteq T$ of odd cardinality.
  \end{enumerate}
\end{lemma}

Deferring the proof of this lemma to \cref{sec:partition-properties},
let us
first show that $(c, \eps)$-degree-balanced cuts always exist in
regular graphs of large enough degree.

\begin{lemma}\label{lem:split}
  For all $c \in [0,1], \eps >0$
  there is a
  $d_0 \in O\big(\frac{c}{\eps^2}\log^2(\frac{c}{\eps^2})\big)$
  such that
  the following holds.
  For every $d > d_0$, every
  $d$-regular graph $G$
  has a $(c, \eps)$-degree-balanced cut.
\end{lemma}

\begin{proof}
  Independently include every vertex $v \in V(G)$ in $S$ with probability
  $c$.
  Let $A_u$ denote the bad event that
  $\Abs{|N(u, S)| - cd} \ge \eps d$. By the Chernoff bound (\cref{thm:chernoff}),
  we have
  \begin{align}
    \prob{A_u} \le 2 \exp(-\eps^2 d/3c) \eqperiod
  \end{align}
  Note that the event $A_u$ depends only on $A_v$ for $v$ within
  distance $2$ of $u$ in $G$, and there are at most $d^2$ many such
  $v$'s.  We want to apply the Lovász local lemma
  (\cref{lem:lll}) to the events $\set{A_v \mid v \in V(G)}$ and
  $x_v = x$ for some parameter $x$.
  The local lemma conditions then require
  $\prob{A_u} \le x(1-x)^{d^2}$ and this right hand side
  is maximized at $x=\frac{1}{d^2+1}$ where, using the bound
  $1-x = 1-1/(d^2+1) \ge e^{-1/d^2}$,
  it becomes
  \[
  x \cdot(1-x)^{d^2} =
  \frac{1}{d^2+1} \cdot \! \left(1-\frac{1}{d^2+1}\right)^{d^2} \ge
  \frac{1}{d^2+1} \cdot \frac{1}{e}.
  \]
  For large enough $d = \Omega(\frac{c}{\eps^2} \log(\frac{c}{\eps^2}))$,
  this is much larger than $\prob{A_u} \le 2\exp(-\eps^2 d/3c)$
  so by \cref{lem:lll} we conclude that
  $\prob{\wedge_{v \in V(G)} \bar{A}_v} > (1-x)^n \ge
  \exp(-\frac{n}{d^2})$.
  All that remains is to argue that there is a positive probability
  that both this happens as well as the size of $S$ being close to~$cn$.
  In particular if
  $\Prob{\Abs{|S| - cn} \ge \eps d} < \prob{\wedge_{v \in
      V(G)} \bar{A}_v}$, the lemma follows.

  By the Chernoff bound (\cref{thm:chernoff}), the cardinality of $S$
  is in $[cn \pm \eps n]$ except with probability at most
  $2 \exp(-\eps^2n/3c)$.  Hence it is sufficient that
  $2 \exp(-\frac{\eps^2n}{3c}) < \exp(-\frac{n}{d^2})$, and for
  $d \gg \sqrt{c}/\eps$ this clearly holds. This concludes the proof.
\end{proof}

With \cref{lem:partition-properties,lem:split} at hand, proving the
\hyperref[lem:partition]{Partition Lemma} simply boils down to
choosing appropriate values for the different constants.

\begin{proof}[Proof of Lemma \ref{lem:partition}]
  Fix $c = 3/4$,
  $\eps = \kappa = 1/16$ and
  $\ell = 7$.
  Let $(S, T)$ be the $(c, \eps)$-degree-balanced cut as
  guaranteed to
  exist in $G$ by
  \cref{lem:split}.
  The cut $(S, T)$ satisfies all the properties of
  \cref{lem:partition-properties}. Hence
  all that remains is to
  verify that the constants were
  chosen appropriately.
  \begin{enumerate}[label=$(\roman*)$]

  \item $G[T]$ is \degrich{1/3}{d/32}: we have that
    $|T| \ge (1 - c - \eps)n = 3n/16$.
    Thus if the graph $G$ is
    \degrich{1/16}{d/32}, the statement follows.
    Observe that for our choice of $\kappa$ and $d$ large enough
    (e.g. $d \ge 2^{16}$ suffices) it
    holds that
    \[
    d \left(
    \kappa -
    2 \sqrt{\frac{1 - \kappa}{\kappa d}}
    \right)
    =
    d \left(
    \frac{1}{16} - 2 \sqrt{\frac{15}{d}}
    \right)
    \ge
    d/32 .
    \]

  \item $G[T]$ is \nonbip{1/2}{1/12, 25}:
    as we may assume that $d \ge 144$,
    this property is satisfied.

  \item $G[T]$ is a $1/3$-expander: the expansion $\alpha$
    guaranteed by \cref{lem:partition-properties} is
    \[
    \alpha =
    \frac{1 - c - 2 \eps}{ 2(1 - c - \eps)} =
    \frac{1/8}{3/8} = 1/3 \eqperiod
    \]
  \end{enumerate}
  The statement follows.
\end{proof}

All that remains is to prove \cref{lem:partition-properties}. In the
following section we recall some results from spectral graph theory
needed for the proof of \cref{lem:partition-properties} which is then
given in
\cref{sec:partition-properties}.

\subsection{Spectral Bounds}

Let us establish some notation and
recall some results from spectral graph theory.

% spectrum graph defs
We denote the adjacency matrix of a graph $G$ by $A_G$ and by $L_G$
its Laplacian $L_G = D_G - A_G$ (where $D_G$ is the diagonal
matrix containing the degrees of the vertices of $G$). For a matrix
$A \in \mathbb{R}^{n \times n}$, denote by
$\lambda_1(A) \le \lambda_2(A) \le \cdots \le \lambda_n(A)$ the
eigenvalues of $A$ in non-decreasing order.

The \emph{edge expansion} of a graph $G$ on $n$ vertices
is
\begin{align}
  \Phi(G) = \min_{\substack{U \subseteq V(G)\\ |U| \le n/2}}
    \frac{|E(U, V(G) \setminus U)|}{|U|} \eqperiod
\end{align}

It is well-known that if the second smallest eigenvalue of the
Laplacian is large, then the graph is a good expander. Note that the
following theorem does not require that $G$ is regular.

\begin{theorem}[\cite{mohar}]\label{thm:exp-edge-ev}
  For all graphs $G$ it holds that
  $\frac{\lambda_2(L_G)}{2} \le \Phi(G)$.
\end{theorem}

\begin{corollary}\label{cor:exp-ev}
  All graphs $G$ have vertex expansion
  $\frac{\lambda_2(L_G)}{2 \Delta(G)}$.
\end{corollary}

\begin{proof}
  As every vertex has at most $\Delta(G)$ neighbors, the neighborhood
  of every set $U$, satisfying $|U| \le n/2$,
  is of size at least $\Phi(G) /\Delta(G)$. The statement follows from
  \cref{thm:exp-edge-ev}.
\end{proof}

Recall that regular random graphs are very good spectral
expanders.
% added words so that the line break is ok.
For the sake of conciseness, let
$\lambda = \max \set{\abs{\lambda_1(A_G)}, \abs{\lambda_{n-1}(A_G)}}$.

\begin{theorem}[\cite{friedman2004proof}]\label{thm:rand-ev}
  Fix $d \ge 3$ and let $nd$ be even. Then, for
  $G \sim \calG(n, d)$ it holds asymptotically almost surely as $n \rightarrow \infty$ that
  $\lambda \le 2\sqrt{d-1} + o(1)$.
\end{theorem}

Another well-known result from spectral graph theory is that the smallest
eigenvalue of the adjacency matrix puts a limit on the maximum size of
an independent set.

\begin{theorem}[Hoffman's bound]\label{thm:hoffman}
  Let $G$ be a $d$-regular graph on $n$ vertices. If $S \subseteq
  V(G)$ is an independent set of $G$, then
  \[
  |S| \le -\frac{n \cdot \lambda_1(A_G)}{d - \lambda_1(A_G)} \eqperiod
  \]
\end{theorem}

\begin{corollary}\label{prop:non-bip}
  Let $G$ be a $d$-regular graph on $n$ vertices.
  For any set $S \subseteq V(G)$ it holds that
  if $|S| > -\frac{2 \cdot n \cdot \lambda_1(A_G)}{d - \lambda_1(A_G)}$,
  then $G[S]$ is not bipartite.
\end{corollary}

\begin{proof}
  For the sake of contradiction
  suppose that there is an $S \subseteq V(G)$ such that $G[S]$ is
  bipartite and
  $|S| > -\frac{2 \cdot n \cdot \lambda_1(A_G)}{d - \lambda_1(A_G)}$.
  Let us denote the partition by $S = A \disjointunion B$.
  W.l.o.g., assume that
  $|A| \ge |S|/2$ and apply \cref{thm:hoffman}
  to $A$ to conclude that
  $
  - \frac{n \cdot \lambda_1(A_G)}{d - \lambda_1(A_G)}
  <
  |A|
  \le
  - \frac{n \cdot \lambda_1(A_G)}{d - \lambda_1(A_G)}
  $.
\end{proof}

Let us recall the mixing
lemma; it states that between linearly sized sets of vertices there
are about as many edges as expected in a random regular
graph.

\begin{lemma}[Expander Mixing Lemma \cite{Hoory06expandergraphs}]
  \label{lem:mix}
  Let $G$ be a $d$-regular graph on $n$ vertices.
  Then for all
  $S, T \subseteq V(G)$:
  \begin{align*}
    \ABS{ \abs{E(S, T)} - \frac{d |S| |T|}{n} }
    \le
    \lambda \sqrt{|S| |T|} \eqperiod
  \end{align*}
\end{lemma}

We also rely on the following theorem that relates the spectrum of the Laplacian
and the existence of a perfect matching.

\begin{theorem}[\cite{BROUWER2005155}]\label{thm:spectrum-matching}
  Let $G$ be a graph on $n$ vertices. If $n$ is even and
  $\lambda_n(L_G) \le 2\lambda_2(L_G)$, then~$G$ has a perfect matching.
\end{theorem}

The following statements consider large induced subgraphs
$H \subseteq G$.  \cref{prop:spectrum} states that if we have good
control of the degrees in $H$, then we have good control
of the spectrum of the Laplacian of $H$ in terms of the spectrum of
the adjacency matrix of $G$.
The proof uses Weyl's
theorem and Cauchy's interlacing theorem, so let us first state these.

\begin{theorem}[Weyl]\label{thm:weyl}
  Let $A, B \in \mathbb{R}^{n \times n}$ be Hermitian.
  Then, for all
  $k \in [n]$,
  \begin{align*}
    \lambda_k(A) + \lambda_1(B) \le \lambda_k(A+B) \le \lambda_k(A) +
    \lambda_n(B) \eqperiod
  \end{align*}
\end{theorem}

\begin{theorem}[Interlacing Theorem]\label{thm:interlace}
  Suppose $A \in \mathbb{R}^{n \times n}$ is symmetric. Let
  $B \in \mathbb{R}^{m \times m}$,
  with $m < n$, be a principal submatrix.
  Then, for all $k \in [m]$,
  \begin{align*}
    \lambda_k(A) \le \lambda_k(B) \le \lambda_{k+n-m}(A) \eqperiod
  \end{align*}
\end{theorem}

\begin{proposition}\label{prop:spectrum}
  Let $G$ be a graph on $n$ vertices
  and $H$ be an induced subgraph of $G$
  with $m$ vertices. Then, for all $k \in [m]$,
  \begin{align*}
    \delta(H) - \lambda_{n-k+1}(A_G) \le \lambda_k(L_H) \le
    \Delta(H) - \lambda_{m-k+1}(A_G) \eqperiod
  \end{align*}
\end{proposition}

\begin{proof}
  By \cref{thm:interlace}, applied to $-A_G$ and $-A_H$, we see that
  for all $k \in [m]$
  \begin{align}
    \lambda_k(-A_G) \le \lambda_k(-A_H) \le
    \lambda_{k+n-m}(-A_G) \eqperiod
  \end{align}
  Note that $\lambda_1(D_H) = \delta(H)$ and $\lambda_m(D_H) = \Delta(H)$.
  Applying \cref{thm:weyl} to $D_H$ and $-A_H$, we conclude
  that, for all $k \in [m]$
  \begin{align}
    \lambda_{k}(-A_G) + \delta(H)
    &\le
    \lambda_k(-A_H) + \lambda_1(D_H)\\
    &\le
    \lambda_k(D_H - A_H)\\
    &\le
    \lambda_{k}(-A_H) + \lambda_m(D_H)
    \le
    \lambda_{k+n-m}(-A_G) + \Delta(H) \eqperiod
  \end{align}
  As $\lambda_k(-A_G) = - \lambda_{n-k+1}(A_G)$ and
  $\lambda_{k+n-m}(-A_G) = - \lambda_{m-k+1}(A_G)$, the statement follows.
\end{proof}

Before commencing with the proof of \cref{lem:partition-properties},
let us state two results that are of non-spectral nature.
The following is a
theorem by Bollobás which captures the distribution of
short cycles in random regular graphs.

\begin{theorem}[\cite{bollobas2001randgraphs}, Corollary 2.19]
  \label{thm:short-cycles}
  Let $d \ge 2$ and $k \ge 3$ be fixed natural numbers and denote by
  $Y_i = Y_i(G)$ the number of $i$-cycles in a graph
  $G \sim \mathcal{G}(n, d)$. Then $Y_3, Y_4, \ldots Y_k$ are
  asymptotically independent Poisson random variables with means
  $\lambda_3, \lambda_4, \ldots, \lambda_k$, where
  $\lambda_i = (d-1)^i/(2i)$.
\end{theorem}

Let us also record a simple observation that establishes that shortest
odd cycles contain no shortcut.

\begin{lemma}\label{lem:non-bip-odd-cycle}
  Let $G$ be any graph and suppose that $C$ is a shortest odd cycle in
  $G$. Then there is no path $p$ connecting two vertices $u,v$ on $C$
  such that both paths of $C$ connecting $u$ to $v$ are longer than
  $p$.
\end{lemma}

\begin{proof}
  Suppose such a path $p$ exists. Let $q \subseteq p$ be a subpath of
  $p$ such that
  \begin{enumerate}[label=$(\roman*)$]
  \item $q$ only shares its endpoints $w_0, w_1$ with $C$, and
  \item the two paths $a_0, a_1$ from
    $w_0$ to $w_1$ on $C$ are longer than $q$.
  \end{enumerate}
  Note that such a subpath $q$ exists as the two paths connecting $u$
  to $v$ on $C$ are both longer than $p$: if no such path $q$ exists,
  then each potential $q$ can be replaced by a part of $C$, thereby
  obtaining a walk from $u$ to $v$ on $C$ of length at most $|p|$; a
  contradiction.
  
  But note that such a $q$ gives rise to a shorter odd cycle: either
  $a_0 \cup q$ or $a_1 \cup q$ is an odd cycle, of length less than
  $C$. This is in contradiction to the initial assumption that $C$ is
  a shortest odd cycle. The statement follows.
\end{proof}

%\subsection{Proof of \cref{lem:partition-properties}}
\subsection{Proof of Lemma \ref{lem:partition-properties}} %Uri
\label{sec:partition-properties}

Recall that by \cref{thm:rand-ev}, with high probability all but the
largest eigenvalue of the adjacency matrix of $G$ are bounded in
magnitude by $2\sqrt{d-1} + o(1)$. In the following we assume that $n$
is large enough such that the $o(1)$ term is small. Let us
argue each property separately.
\begin{enumerate}[label=$(\roman*)$]
\item
  Let $U \subseteq V(G)$ be any set of size $\kappa n$.
  Apply the mixing lemma (\cref{lem:mix}) to the graph~$G$ to conclude
  that
  \[
  |E(U, V(G) \setminus U)| \le
  \kappa n \cdot d
  \big(
  (1 - \kappa) +
  2\sqrt{\frac{1-\kappa}{\kappa d}}
  + o(1)
  \big) \eqperiod
  \]
  As $G$ is a $d$-regular graph, we
  conclude that the average degree in $G[U]$ is at least
  $
  d (\kappa - 2\sqrt{\frac{1-\kappa}{\kappa d}} - o(1))
  $. By the observation that if the average degree is at least $t$,
  then there is a vertex of degree at least $\ceil{t}$, the statement
  follows for $n$ large enough.

\item
  Recall that a sum of independent Poisson variables $X_1, \ldots, X_k$
  with means $\mu_1, \ldots, \mu_k$
  is again a Poisson variable with mean $\sum_{i \in [k]} \mu_i$.
  Hence the number of cycles in $G$ of length at most
  $\ell$ is, according to \cref{thm:short-cycles}, a Poisson
  random variable $Y$ with mean
  \begin{align}
    \mu
    =
    \sum_{i = 3}^{\ell} \frac{(d-1)^i}{2i}
    \le
    d^{\ell}/6 \eqcomma
  \end{align}
  where we used that $d^{\ell - 1} + (d - 1)^\ell \le d^\ell$.
  \Cref{thm:poi-tail} then tells us that for any $\gamma > 0$,
  independent of $n$, it holds that
  \begin{align}
    \prob{Y \ge \gamma \log n}
    \le e^{-\mu} \left(\frac{e \mu}{\gamma\log n}\right)^{\gamma\log n}
    < \frac{1}{n}
    \eqcomma
  \end{align}
  where the strict inequality holds for $n$ large enough.  Hence we
  may assume that $Y < \gamma \log n$. Let $S \subseteq V(G)$ be a set
  of vertices that contains one vertex from each cycle of length at
  most $\ell$. By assumption $|S| < \gamma \log n$ and the shortest
  cycle in $G \setminus S$ is of length at least $\ell$.

  We also know that all but
  the largest eigenvalue of
  the adjacency matrix of $G$
  are bounded in magnitude by
  $2\sqrt{d-1} + o(1)$.
  Apply
  \cref{prop:non-bip}
  to conclude that
  no subset
  $W \subseteq V(G)$
  of size at least
  $|W| \ge 5 \frac{n}{\sqrt{d}}$
  induces a
  bipartite subgraph, in other words
  any such $G[W]$ contains an odd cycle.

  Let $U \subseteq V(G)$ be of size at least
  $6\frac{n}{\sqrt{d}}$. For $n$ large, it holds that
  $G[U \setminus S]$ is of size at least $5\frac{n}{\sqrt{d}}$ and
  thus contains an odd cycle of length at least $\ell$. Let $C$ denote
  such a cycle.
  What remains is to show that there is a cycle in $G[U]$ that is
  simultaneously of length at least $\ell$ as well as bounded in
  length by $3\cdiam[\beta/2] \log n$.

  Towards contradiction suppose that
  $|C| \ge 3\cdiam[\beta/2] \log n$ and that $C$ is a shortest
  odd cycle in $G[U\setminus S]$. Arbitrarily split the cycle $C$ into
  four paths $A_1, B_1, A_2$ and $B_2$, such that $B_1$ and $B_2$
  separate $A_1$ from $A_2$ on $C$, both $A_i$s are of size at least
  $\frac{1}{4}\cdiam[\beta/2]\log n$, and both $B_i$s are of size at
  least $\frac{9}{8}\cdiam[\beta/2]\log n$.

  We may assume that $\gamma \le \frac{\beta}{9} \cdiam[\beta/2]$ so
  that for $n$ large we can apply \cref{cor:exp-diam-remove} to
  $G[U]$, $S, A_1$ and $A_2$ to conclude that in $G[U\setminus S]$
  there is a path $p$ connecting $A_1$ to $A_2$ of length at most
  $\cdiam[\beta/2] \log n$. This contradicts
  \cref{lem:non-bip-odd-cycle} as both paths of $C$ connecting $A_1$ to
  $A_2$ are of length at least
  $\frac{9}{8}\cdiam[\beta/2] \log n$.

  We conclude that there is an odd cycle of length at most
  $3\cdiam[\beta/2] \log n$ in $G[U \setminus S]$. As there are
  no cycles of length at most $\ell$ in $G[U \setminus S]$ we see that
  this cycle is also of length at least $\ell$, as required.
  
\item
  Applying \cref{prop:spectrum} to $G$ and $G[T]$,
  we see that
  \begin{align}
    \lambda_2(L_{G[T]}) \ge
    \delta(G[T]) - \lambda_{n-1}(A_G) \eqperiod
  \end{align}
  Every vertex $v \in T$ has degree at least
  $(1 - c - \eps)d$ in
  $G[T]$.
  Furthermore,
  as $\lambda_{n-1}(A_G)$
  is bounded by
  $2\sqrt{d-1} + o(1)$
  and we assumed that
  $d \ge 4/\eps^2$,
  we obtain that
  $\lambda_2(L_{G[T]}) \ge (1 - c - 2\eps)d$.
  Applying \cref{cor:exp-ev},
  we conclude that $G[T]$ has vertex
  expansion at least $\frac{1-c-2\eps}{2(1-c-\eps)}$.

\item
  Let $U \subseteq T$ of odd cardinality be as in the statement, and
  denote by $m$ the number of vertices in
  $G \setminus U$.
  By \cref{thm:spectrum-matching}, it is sufficient to establish the
  bound
  $\lambda_m(L_{G \setminus U}) \le 2 \lambda_2(L_{G \setminus U})$
  on the eigenvalues of the Laplacian of
  $G \setminus U$.
  Applying \cref{prop:spectrum}
  to $G \setminus U$, we can bound these eigenvalues in
  terms of the eigenvalues
  of the adjacency matrix of $G$, obtaining
  \begin{align}
    \lambda_{m}(L_{G \setminus U})
    &\le d - \lambda_{1}(A_G) \text{~and}\\
    \lambda_2(L_{G \setminus U})
    &\ge (c-\eps)d - \lambda_{n-1}(A_G) \eqperiod
  \end{align}
  As $\lambda_1(A_G)$ and $\lambda_{n-1}(A_G)$ are both bounded in
  absolute value by $2\sqrt{d-1} + o(1)$ we thus conclude
  \[
  2 \lambda_{2}(L_{G \setminus U}) - \lambda_{m}(L_{G \setminus U})
  \ge
  (2(c-\eps)-1) d - 2\sqrt{d-1} - o(1).
  \]
  Since $c > 1/2+\eps$
  and we assumed that
  $d \ge (c-1/2-\eps)^{-2}$, we have that
  $
  \lambda_{m}(L_{G \setminus U})
  \le
  2 \lambda_{1}(L_{G \setminus U})
  $
  as desired.
\end{enumerate}

\section{Embedding Theorem}
\label{sec:odd-minor}

In this section we prove our embedding theorem (\cref{thm:odd-minor}).
Before starting with the proof, let us establish some notation and
recall some facts from graph theory.

\subsection{Further Graph Theory Preliminaries}

% Balanced Separator
In a graph $G = (V, E)$ on $n$ vertices a vertex set $S \subseteq V$
is a \emph{balanced separator in $G$} if there is a partition
$V = A \disjointunion B \disjointunion S$ of the vertex set of $G$
such that $|A|, |B| \le 2n/3$, and $G$ has no edges between $A$ and
$B$.

Large vertex expansion implies that balanced separators are large:
the next lemma makes this well-known connection precise.

\begin{lemma}\label{lem:exp-sep}
  Let $G$ be an $\alpha$-expander on $n$ vertices, and let $S$ be a
  balanced separator in $G$.
  Then $|S| \ge \frac{\alpha n}{3(1 + \alpha)}$.
\end{lemma}

\begin{proof}
  Let $S$ be a balanced separator in $G$ of size
  $|S| = s$, separating $A$ and $B$, with
  $|A| = a$,
  $|B| = b$. Without loss of generality assume that
  $a \le b \le 2n/3$.
  Clearly, $a + s \ge n/3$. Further,
  $N(A) \subseteq S$, and since $a \le n/2$,
  by expansion, we get that $s \ge \alpha a$. In other words,
  $s/\alpha \ge a$, which when substituted into
  $a + s \ge n/3$ yields $s(1+1/\alpha) \ge n/3$.
\end{proof}

We also require the following lemma on vertex-disjoint paths in
expanders.

\begin{lemma}[\cite{fm2019cycle}]\label{lem:exp-path-disj}
  Let $G=(V, E)$ be an $\alpha$-expander and let $A, B \subseteq V$ be
  two vertex sets of sizes $|A|, |B| \ge t$ for some $t > 0$. Then $G$
  contains at least $\frac{t\alpha}{1+\alpha}$ vertex-disjoint paths
  between $A$ and $B$.
\end{lemma}

% The rest of this section is devoted to the
% proof of \cref{thm:odd-minor}.

\subsection{Proof of \cref{thm:odd-minor}}
\label{sec:odd-minor-proof}

We now proceed with the proof of \cref{thm:odd-minor}, restated here
for convenience.

\EmbeddingTheorem*

When embedding a high degree vertex $x \in V(H)$ into $G$, we want to find a vertex
$v \in V(G)$ of high degree such that many neighbors are connected to large,
disjoint sets of vertices. These large sets are very useful as they
guarantee that there are many vertices to which we can connect a vertex embedding.
The following definition makes this intuition precise.

\begin{definition}[Cross]\label{def:cross}
  An \emph{\cross{(r, s)}} in a graph $G = (V, E)$
  is a tuple $(v, \branches)$,
  where $v \in V$ is a vertex and
  $\branches \subseteq 2^{V}$ consists of
  $r$ pairwise disjoint vertex sets
  $U \subseteq V \setminus \set{v}$,
  each of size $|U| = s$,
  such that $N(v) \cap U \neq \emptyset$
  and the graph $G[U]$ is connected.
  We refer to $v$ as the \emph{center} of the cross and to
  $\branches$ as the \emph{branches} of the cross.
\end{definition}

The following lemma shows that crosses always exist in expanders with
sufficiently large maximum degree.

\begin{restatable}{lemma}{CrossLemma}\label{lem:cross}
  For all $\beta > 0$ and
  $\gamma = \frac{\beta}{3 (1 + \beta)}$
  the following holds.
  Let $G$ be a $\beta$-expander on $n$ vertices that is
  \degrich{1-2/k}{(1+1/\beta)r},
  for some $k \ge 3$  and $r > 0$ such that
  $r \le \frac{\gamma^3 n}{k(1+\gamma)}$.
  Then $G$ contains an \cross{(r,s)}, for
  all $s$ that satisfy
  $r \cdot s \le \frac{\gamma^2 n}{k(1+\gamma)}$.
\end{restatable}

The proof is an adaptation of a proof by Krivelevich and
Nenadov~\cite{kn2018MinorsNoCut} and is deferred to \cref{sec:cross
  proof}.  We also have the following lemma which is what allows us to
choose the path length parities in the ``furthermore'' part of
\cref{thm:odd-minor}.
It states that if there is an odd cycle in the graph,
then there is an odd and even path between any vertex $u$ and a large enough
set~$A$ of vertices. Note that this does not necessarily hold if $A$ is too small:
the vertex $u$ may have degree~$1$ and~$A$ may be
the single neighbor of $u$.  Similarly a lower bound on the length of the odd
cycle is needed.

\begin{restatable}{lemma}{OddEvenPath}\label{lem:odd-even-path}\label{LEM:ODD-EVEN-PATH}
  For all $\beta > 0$ the following holds.
  Let $G$ be a $\beta$-expander on
  $n$ vertices that contains an
  odd cycle of length $\ell \ge 1 + 2/\beta$. Then,
  for all $u \in V(G)$ and $A \subseteq V(G)$, of size
  $|A| \ge (\cdiam[\beta] \log n + 1)(1 + 2/\beta)$,
  there is a vertex $v \in A$ such that
  $u$ and $v$ are connected by both an odd and an even path, each of length
  at most $(15 \cdiam[\beta/2]/\beta) \log{n} + \ell$.
\end{restatable}

We defer the proof of \cref{lem:odd-even-path} to
\cref{sec:odd-even-path}.

We now prove \cref{thm:odd-minor} with the assumption of
odd-cycle-robustness. Furthermore, the proof makes all paths of odd
length, though it is immediate that one can choose the parities. To
get the theorem without the assumption of odd-cycle-robustness, one
just has to replace the application of \cref{lem:odd-even-path} by any
shortest path (which, by \cref{lem:exp-diam} is short).

The main idea is due to Krivelevich and
Nenadov~\cite{kn2018MinorsNoCut}
(see also \cite{krivelevich2018expanders}).
In contrast to their work we cannot
directly embed the vertices into the graph but rather take a
detour by embedding appropriately sized crosses for each vertex and
then
connect branches of crosses that correspond to embeddings of
adjacent vertices. The reason for this difference is that the present
theorem deals with topological minors rather than plain graph
minors (the difference is that in topological minors vertices
are connected by vertex disjoint paths while in graph minors
subgraphs are connected).

In order for this to work we need to make some further changes to the
embedding process used.  In their work, three sets of vertices are
maintained throughout the process: one set $A$ of ``discarded'' vertices,
one set $B$ of vertices used in the embedding, and the remaining set $C$ of
vertices.  A key invariant which is maintained is that the set of discarded
vertices expand poorly into the set of remaining vertices, which
together with expansion implies that not too many vertices can be
discarded.  In our case, some of the discarded vertices may in fact
have good expansion into $C$, but we can maintain the property that
there are not too many such vertices. The details are worked out in
what follows. If the verbal description is ambiguous, there is an
algorithmic description in \cref{sec:pseudocode} (\cref{alg:embed-graph}).

Formally, the algorithm maintains a partition
$A \disjointunion A' \disjointunion \dot\bigcup_{B \in \calB}B
\disjointunion C$ of the vertices of $G$.  The sets $A$, $B$, and $C$
play the same roles as in the informal description above, and $A'$ is
an additional set of discarded vertices which may have large expansion
into $C$.
When the algorithm terminates,
every vertex $v \in V(H)$ (edge $e \in E(H)$, respectively)
has a vertex embedding
$\embed{v} \in \calB$ (an edge embedding $\embed{e} \in \calB$)
giving a topological minor of $H$ in $G$.
Initially, all sets except $C = V(G)$ are empty.

Let $\beta = \frac{\alpha}{3(1+\alpha)}$ be the constant from
\cref{lem:exp-sep} for the lower bound on the size of a balanced
separator in an $\alpha$-expander.  At several points in the algorithm
we want to ensure that $G[C]$ is a $\beta$-expander. This is achieved
by removing any subset $U \subseteq C$ of size $|U| \le |C|/2$ with
small neighborhood $|N(U, C \setminus U)| < \beta |U|$ from $C$ and
adding it to $A$ (i.e., letting $C \leftarrow C \setminus U$ and $A
\leftarrow A \cup U$). Clearly once there are no sets $U \subseteq C$
left as above, $G[C]$ is a $\beta$-expander.

Throughout the algorithm the following invariants are maintained:
\begin{enumerate}[label=$(\roman*)$]
\item
  $C$ never increases in size and $|C| \ge n(1-2/k)$,
\item
  $G[C]$ is a $\beta$-expander (by restoring expansion as
  described above whenever needed),
\item
  $N(A, C) < \beta |A|$, and
\item
  $|A'| < \beta |A|/2$.
\end{enumerate}
The algorithm
maintains the set $I \subseteq V(H)$ to keep track of the vertices
already embedded.

\newcommand{\constPathLen}{\big(18 \cdiam[\beta/2]/\beta\big) }
Let
$r(d) = d(1 + 4/\beta) - 1 < 25d/\alpha$ and
$
  s = \constPathLen \log n
  $.
In what follows we assume $\eps$ is sufficiently small as a function of $\beta$.

Fix a vertex $x \in V(H) \setminus I$ not already embedded and
apply \cref{lem:cross} to $G[C]$
to obtain a \cross{(r(\deg_H(x)),s)} $\embed{x}$.
Remove $\embed{x}$ from
$C$ and add it to $\calB$ as the vertex embedding of $x$ (set $C \leftarrow C \setminus \embed{x}$ and
$\calB \leftarrow \calB \cup \set{\embed{x}}$), and restore
$\beta$-expansion in $G[C]$.

Let us check that all the conditions of
\cref{lem:cross} are satisfied.
First, we need that $G[C]$ is \degrich{1-2/k}{(1+3(1+\beta)/\beta)r(\deg_H(x))}.  We
have $1+3(1+\beta)/\beta \le 22/\alpha$ and thus
\begin{align*}
  r(\deg_H(x))(1 + 3(1+\beta)/\beta)
  &\le
    r(\Delta(H)) \frac{22}{\alpha}
  <
    \frac{550}{\alpha^2} \Delta(H) .
\end{align*}
Furthermore since $G$ is \degrich{1-4/k}{550\Delta(H)/\alpha^2} and
$|C| \ge (1-2/k)n$, $G[C]$ is
\degrich{1-2/k}{550\Delta(H)/\alpha^2}.
Second we need to check that
\begin{align*}
  r(\deg_H(x)) &\le \frac{\gamma^3 |C|}{k(1+\gamma)} &&\text{and} &
  r(\deg_H(x)) \cdot s & \le
  \frac{ \gamma^2 |C|}
       {k (1+\gamma)} \eqcomma
\end{align*}
where $\gamma = \frac{\beta}{3(1+\beta)}$.
Since $|C| \ge (1-2/k)n$ and $\Delta(H) \le |V(H)| \le \frac{\eps n}{k \log
  n}$ the first bound clearly holds for $n$ large enough, and provided $\eps$ is sufficiently
small as a function of $\alpha$ the second bound also holds.
Thus we can indeed apply \cref{lem:cross} on
$G[C]$ with the desired choice of $r$ and $s$.

%   The short calculation:
%   \begin{align*}
%     r(\deg_H(x)) \cdot s
%     &<
%       \Delta(H)(1 + 4/\beta) \cdot s\\
%     &\le
%       \frac{\eps n}{\log n}(1 + 4/\beta) \cdot s\\
%     &\le
%       \frac{\beta^3(1 + 4/\beta)n}{260 \cdot k}\\
%     &\le
%       \frac{\beta^2n}{52 \cdot k}\\
%     &\le
%       \frac{\beta^2 n}{9 \cdot (1 + 1/3)(3+4/3)k}\\
%     &\le
%       \frac{\beta^2 n}{9 \cdot (1 + \beta)(3+4\beta)k}\\
%     &\le
%       \frac{\left(\frac{\beta}{3(1+\beta)}\right)^2 (1 - 2/k)n}
%       {\left(\frac{3 + 4 \beta}{3(1+\beta)}\right)k}
%   \end{align*}

After embedding $x$, we need to connect the embedding
$\embed{x}$ to the embeddings of the
neighbors $N_H(x) \cap I = \set{y_1, \ldots, y_\nu}$ that are already
embedded.
Suppose, for now, that the vertex embeddings have
branches $\branch_x \in \embed{x}$ and $\branch_{y_i} \in \embed{y_i}$ that are
$\beta$-expanding into $C$
(i.e.\ $|N(\branch_x, C)|, |N(\branch_{y_i}, C)| \ge \beta s$), and such that neither of the two
branches are already used to connect $x$, resp. $y_i$, to a neighbor.

By the assumption on odd-cycle-robustness, we see that $G[C]$ is
non-bipartite and contains an odd cycle $c$ of length
\begin{align}
  1 + 2/\beta
  \le
  |c|
  \le
  3 \cdiam[\beta/2] \log n \eqperiod
\end{align}
As each branch is rather large, of size $s$, we can
apply \cref{lem:odd-even-path} to $G[C]$, $N(\branch_{x}, C)$ and
$N(\branch_{y_i}, C)$ to
conclude that
in $G[C]$ there is an odd
path $q_i$ connecting $\branch_x$ to $\branch_{y_i}$ of length
$\constPathLen \log n \le s$.
Remove $q_i$ from $C$, add it to $\calB$ as the edge embedding
$\embed{\set{x, y_i}}$ and restore $\beta$-expansion in $G[C]$. This process is
illustrated in \cref{fig:odd-minor-edge} and can be found as pseudo
code in \cref{alg:embed-graph}.

%% Applying \cref{lem:odd-even-path} only requires that one of $U_x$, $U_{y_i}$ is large.

\begin{figure}
  \centering
  \includegraphics{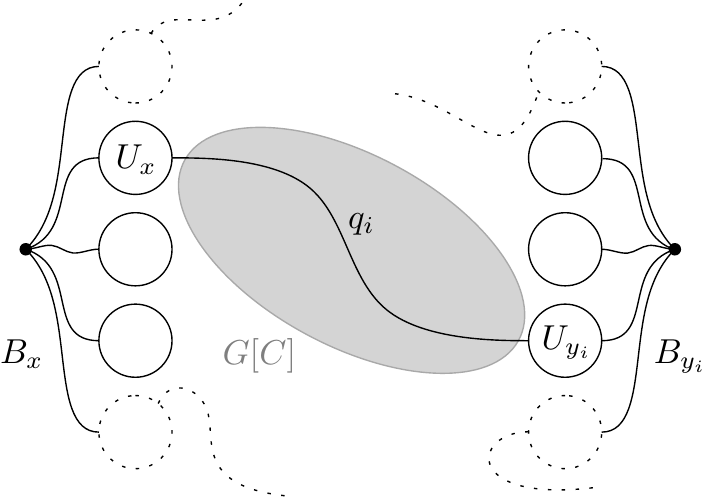}
  \caption{The vertex embedding $\embed{x}$ is connected to
    $\embed{y_i}$ by the path $q_i$ which connects the two branches
    $\branch_x$ and $\branch_{y_i}$. The dotted branches have an
    edge embedding adjacent and can thus not be used to connect
    $\embed{y_i}$ to $\embed{x}$.}
  \label{fig:odd-minor-edge}
\end{figure}

If all branches of a vertex embedding $\embed{z}$ have either too few
neighbors in $C$ or are already adjacent to an edge embedding
(i.e., have already been used to embed some other edge), then we want
to remove the embedding of $z$. This has to be done in a careful
manner in order not to break the invariants.
First, move all branches that
are not used to connect $z$ to a neighbor to
$A$. Note that each such branch $\branch$ satisfies
$|N(\branch, C)| < \beta |\branch|$. Next, move the remaining
branches along with the adjacent edge embeddings to $A'$. Last,
the center of $\embed{z}$ is moved to $A'$ and $z$ is removed from $I$.
Note that at most $2(\deg_H(z) - 1)s$
many vertices are moved to
$A'$: at most $\deg_H(z) - 1$ many branches of size $s$
and as many edge embeddings, each again
of size at most~$s$.
On the other hand at least
\begin{align}
  \big(r(\deg_H(z)) - (\deg_H(z) - 1)\big) \cdot s
  &= \deg_H(z) \cdot 4 s/\beta
\end{align}
many vertices are moved to $A$. Hence the invariant
$|A'| < \beta |A|/2$ is maintained.

The algorithm terminates the first time either $I = V(H)$ or
$|A| \ge n/k$.
This completes the description of the algorithm.

It remains to argue that it cannot happen that $|A| \ge n/k$,
in other words that when the algorithm terminates, all of $H$ is embedded in $G$.
To this end,
observe that the size of $\cup_{B \in \calB} B$ is upper bounded by
\begin{align*}
  s \cdot \left(
  |E(H)| +
  \sum_{v \in V(H)} r(\deg_H(v))
  \right)
  &<
  s \cdot \left(
    |E(H)| +
    (4/\beta + 1) \sum_{v \in V(H)} \deg_H(v)
    \right)
    \\
  &\le s \cdot |E(H)| \cdot \frac{11}{\beta}\\
  &\le s \cdot \frac{\eps n}{k \log n} \cdot \frac{11}{\beta} \\
  &\le \beta n/2k \eqperiod
\end{align*}
Furthermore, while $|A| \le n/k$ we have that
\begin{align*}
  |A'|
  &< \beta |A|/2
  \le \beta n/2k \eqperiod
\end{align*}
Note that this also holds
the first time $|A|$ becomes larger than $n/k$.
%%% using that vertices are first moved to A, then to A'.
This shows, in particular, that the invariant $|C| \ge n(1 - 2/k)$
is maintained throughout the execution of the algorithm.

For the sake of contradiction, suppose that
the algorithm terminates because of $|A| \ge n/k$.
Note that $|N(A)| \le |A'| + \left| \cup_{B \in \calB} B \right| + |N(A, C)| < \beta (|A| + n/k)$.
We do a case distinction, depending on the size of
$A$. In both cases we derive contradiction and thus show
that the algorithm only terminates after having embedded all of $H$ into $G$.
\begin{enumerate}[label=Case \arabic*:, leftmargin=*]
\item $n/k \le |A| \le n/2$.
  By expansion and using $\beta < \alpha/3$ we have
  \[
  \alpha |A|
    \le
    |N(A)|
    <
    \beta (|A| + n/k) < \frac{\alpha}{3}(|A| + n/k) \eqcomma
    \]
  which together with $|A| \ge n/k$ yields the desired contradiction.
\item $|A| > n/2$.
  Note that the first time
  $|A| \ge n/k$, it also holds that
  $|A| \le n(1 + 1/k)/2$ as
  the sets added to $A$ are of size at most
  $|C|/2 \le (n - |A|)/2$.
  Hence we get that
  \begin{align*}
    |N(A)|
    &< \beta \big(n/k + |A|\big)
    < \beta n\big(1/k + (1 + 1/k)/2\big)
    \le \frac{\alpha n}{3(1 + \alpha)} \eqcomma
  \end{align*}
  using that $k \ge 3$.
  %%% An expander cannot be a collection of cycles.
  Note that $N(A)$ is a balanced separator, separating $A$ from
  $V(G) \setminus A \setminus N(A)$. But this is a contradiction, since \cref{lem:exp-sep}
  states that any balanced separator of $G$ has size at least $\frac{\alpha n}{3(1+\alpha)}$.
\end{enumerate}

\subsection{Crosses in Expanders}
\label{sec:cross proof}

Let us now turn to the proof of \cref{lem:cross}, restated here for convenience.

\CrossLemma*

The proof follows a similar algorithm as the proof of \cref{thm:odd-minor}.
In this case we can in fact more or less use the original argument of
Krivelevich and Nenadov \cite{kn2018MinorsNoCut} without any extensions.

\begin{proof}

  The high-level idea of the proof is as follows.  First, using the
  embedding argument of Krivelevich and Nenadov, we find some number
  $r' > r$ pairwise disjoint sets $B_1, \ldots, B_{r'}$ of vertices of
  $G$ and a final set $C$ disjoint from all $B_i$s such that (i) each
  $B_i$ is a connected subgraph of $G$ on $s$ vertices, (ii) the
  $B_i$s have many neighbors in $C$, and (iii) $G[C]$ is expanding.
  Having these subsets, we can then choose a representative
  $u_i \in N(B_i, C)$ of each $B_i$, take a vertex $v \in C$ of high
  degree (which exists by the max-degree-robustness of $G$), and apply
  \cref{lem:exp-path-disj} to find vertex-disjoint paths connecting
  $N(v)$ to the $u_i$s.  This establishes the existence of a cross
  with $v$ as the center and the $B_i$s together with the respective
  paths as branches.  See \cref{fig:odd-minor-vertex} for an
  illustration.

  Let us proceed with the details.
  In case there is some ambiguity in the
  verbal description there is also a pseudo code description in
  \cref{sec:pseudocode} of what follows.

  \begin{figure}
    \centering
    \includegraphics{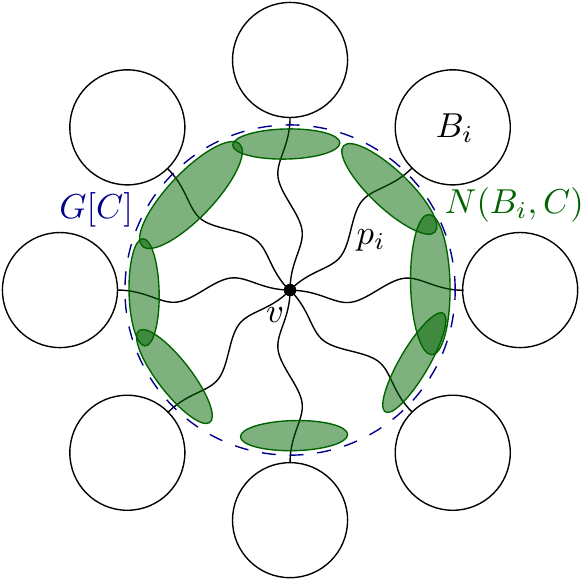}
    \caption{
      A cross with center $v$ and
      branches $\setdescr{V(p_i) \cup B_i}{i \in [r]}$.
      }
    \label{fig:odd-minor-vertex}
  \end{figure}

  Fix $r$, set $r' =  r(1 + 1/\gamma)$ and choose $s \in \mathbb{N}$ maximal such that
  $s \le \frac{\gamma n}{k \cdot r'}$. Note that
  $s \ge 1/\gamma$ and
  if the statement holds for this maximal $s$,
  then it also
  holds for smaller values of $s$,
  as one can always shrink the branches to the appropriate size.

  Let us describe an algorithm to identify the sets
  $\calB = \setdescr{\embed{i} \subseteq V(G)}{i \in [r']}$.
  The algorithm maintains a partition
  $A \disjointunion  \dot\bigcup_{B \in \calB} B \disjointunion C$ of the
  vertices of $G$.
  Initially, all sets except $C = V(G)$ are empty.
  After running the procedure, the set $\calB$
  contains $r'$ pairwise vertex-disjoint sets
  such that for each $B_i \in \calB$ it holds that
  $|B_i| = s$ and the induced subgraph
  $G[B_i]$ is a single connected component. Further,
  for all subfamilies $\mathcal{F} \subseteq \calB$ it
  holds that
  $\left|\bigcup_{F \in \mathcal{F}}N(F, C)\right| \ge \gamma s |\mathcal{F}|$.
  Throughout the execution of the
  algorithm the following invariants are maintained
  \begin{enumerate}[label=$(\roman*)$]
  \item
    $C$ never increases in size and $|C| \ge n(1-2/k)$,
  \item
    $G[C]$ is a $\gamma$-expander (by restoring expansion whenever needed),
  \item
    $N(A, C) < \beta |A|$, and
  \item
    $\left|\dot\bigcup_{B \in \calB} B\right| \le r' \cdot s \le \gamma n/k$.
  \end{enumerate}
  The algorithm terminates if $\calB$ contains $r'$
  vertex sets as described,
  or if the size of $A$ reaches $|A| \ge n/k$.
  The latter case can only occur if there is a small
  balanced separator in $G$.
  But $G$ is a $\beta$-expander, so we know from
  \cref{lem:exp-sep} that there are no small balanced separators
  and hence when the algorithm terminates, $\calB$ must contain $r'$ sets
  as described above.

  Like in the main algorithm used in the proof of \cref{thm:odd-minor},
  we want to ensure that $G[C]$ is a
  $\gamma$-expander throughout the algorithm, which is achieved by removing any subset
  $U \subseteq C$ of size $|U| \le |C|/2$
  with small neighborhood $|N(U, C \setminus U)| < \gamma |U|$
  from $C = C \setminus U$ and
  adding it to $A = A \cup U$.

  Repeat the following while there are less than $r'$ sets in $\calB$.
  Choose a set of vertices $U \subseteq C$ of size $|U| = s$ such
  that $G[U]$ is a single connected component. Remove this set from
  $C = C \setminus U$, add it to $\calB = \calB \cup \set{U}$ and restore expansion
  in $G[C]$.
  After expansion is restored,
  let $\mathcal{F} \subseteq \calB$ be a maximal (possibly empty) family % of sets
  such that
  $\left|\bigcup_{F \in \mathcal{F}} N(F, C) \right| <
  \gamma s|\mathcal{F}|$.
  Remove $\mathcal{F}$ from $\calB = \calB \setminus \mathcal{F}$, and add
  these sets to $A = A \cup_{F \in \mathcal{F}} F$.

  As mentioned before, the
  algorithm terminates once there are either $r'$ sets in $\calB$ or the
  set $A$ is large $|A| \ge n/k$. This completes the description of
  the algorithm.
  Let us argue that the latter cannot happen -- for the sake of
  contradiction, suppose the algorithm terminates because
  $|A| \ge n/k$.
  Note that we have $|N(A)| \le |\cup_{B \in \calB} B| + |N(A, C)| < \gamma (|A| + n/k)$.
  We do a case distincion on the size of $|A|$.
  \begin{enumerate}[label=Case \arabic*:, leftmargin=*]
  \item $n/k \le |A| \le n/2$.
    By expansion, $\beta |A| \le N(A) < \gamma (|A| + n/k) < \frac{\beta}{3} (|A| + n/k)$.
    As $|A| \ge n/k$ this is a contradiction.
  \item $|A| \ge n/2$.
    Note that the first time
    $|A| \ge n/k$, it also holds that
    $|A| \le n(1 + 1/k)/2$ as
    the sets added to $A$ are of size at most
    $|C|/2 \le (n - |A|)/2$.
    % also using that $s$ is much smaller.
    Hence we get (using $k \ge 3$) that
    \begin{align*}
      |N(A)| & \le \gamma (n/k + |A|) \le \gamma n = \frac{\beta n}{3(1+\beta)}.
    \end{align*}
  Note that $N(A)$ is a balanced separator, separating $A$ from
  $V(G) \setminus A$. But this is a contradiction, since \cref{lem:exp-sep}
  states that any balanced separator of $G$ has size at least $\frac{\beta n}{3(1+\beta)}$.
  \end{enumerate}

  It remains to obtain an \cross{(r, s)} from the sets $B_i$ and the
  remaining part $C$.  Choose a vertex $v \in C$ of degree at least
  $\deg_{G[C]}(v) \ge r'$.  Such a vertex $v$ exists, as
  $|C| \ge (1 - 2/k)n$ is large (first invariant) and the statement
  assumes that there is a vertex of degree $r'$ in every induced
  subgraph of size at least $(1 - 2/k)n$. Let $T$ be a transversal of
  the family $\setdescr{N(B, C)}{B \in \calB}$. Note that such a
  transversal $T$ exists by Hall's marriage theorem, using that
  $s \ge 1/\beta$ and that every subset of $\calB$ is
  $\beta$-expanding into $C$.
  
  Apply \cref{lem:exp-path-disj} to $G[C]$ and the vertex sets $N(v)$
  and $T$ to conclude that there are pairwise vertex-disjoint paths
  $\setdescr{p_i}{i \in [r]}$ each connecting $N(v)$ to a set
  $N(B_i, C)$, for some $B_i \in \calB$. We let the \cross{(r,s)} have
  center $v$ and branches
  $\setdescr{V(p_i) \cup \embed{i}}{i \in [r]}$. Let us verify that
  this is indeed a valid \cross{(r, s)}.

  Each path $p_i$ connects $N(v)$ to $N(B_i, C)$ and we thus have
  that, as required, each branch intersects $N(v)$ and that the
  branches are connected, where we use that the sets $B_i \in \calB$
  are by definition connected.  We also need to verify that the
  branches are pairwise vertex-disjoint. To this end recall that the
  sets $B_i \in \calB$ are pairwise disjoint and, furthermore, each
  such set is disjoint from $C$. As the pairwise vertex-disjoint paths
  $p_i$ live in $G[C]$, these paths do not intersect
  $\cup_{i \in [r]}B_i$ and we may thus conclude that the branches are
  pairwise vertex-disjoint. Finally, we also need to check that each
  branch is of size $s$: each set $B_i$ is of size $s$ and thus each
  branch is of size at least $s$. Shrinking the branches to the
  appropriate size recovers the statement.
\end{proof}

\subsection{Odd and Even Paths}
\label{sec:odd-even-path}

In this section we prove \cref{lem:odd-even-path}.

\OddEvenPath*

The lemma is a corollary of a more general statement about short paths
in $\alpha$-expanders.
The lemma states that if sets $S, T$,
where $|S| \gtrsim |T|/\alpha$, are
connected by $|T|$ many short vertex-disjoint paths,
then for any large set $U$ there is again a set of short
vertex-disjoint paths that does not only connect every vertex of $T$ to
$S$ but also a vertex from $U$ to $S$.

In order to state the lemma, let us introduce some notation.
For a graph $G$ and vertex sets $S, T \subseteq V(G)$,
denote by $\totallength[G]{T}{S}$ the
minimum total length of connecting all vertices of
$T$ to $S$ by pairwise
vertex-disjoint paths;
\begin{align}
  \totallength[G]{T}{S} &=
  \min_{\setdescr{p_t}{t \in T}} \sum_{t \in T}|p_t|
\end{align}
where $\setdescr{p_t}{t \in T}$
ranges over all sets of pairwise vertex-disjoint paths such that $p_t$
connects $t$ to $S$
(note the paths $\{p_t\}$ are not allowed to intersect even in $S$).
If no such set of paths exists,
the value of the minimum is taken to be $\infty$.
If the graph $G$ is clear from context, we omit the superscript.

A similar lemma (though without the essential upper bound on the path
lengths) has appeared in e.g. \cite{fm2019cycle}.

\begin{lemma}\label{lem:exp-paths}\label{LEM:EXP-PATHS}
  % For all $\beta > 0$
  % the following holds.
  Let $G$ be a $\beta$-expander on
  $n$ vertices and
  $S, T\subseteq V(G)$ satisfy
  $|S| \ge |T|(1 + 2/\beta)$.
  Then every set $U \subseteq V(G)$, of size
  $|U| \ge  (\totallength{T}{S} + |T|)(1 + 2/\beta)$,
  contains a vertex $u \in U$ such that
  $\totallength{T \cup \set{u}}{S} \le
  7(\totallength{T}{S} + |T|)/\beta + 2 \cdiam[\beta/2] \log n $.
\end{lemma}

\cref{lem:odd-even-path} follows by a single application of
\cref{lem:exp-paths}.

% \begin{proof}[Proof of \cref{lem:odd-even-path}]
\begin{proof}[Proof of Lemma~\ref{lem:odd-even-path}] %Uri
  Let $C$ denote an odd cycle of length $\ell \ge 1 + 2/\beta$, as
  guaranteed to exist, and denote by $p$ a shortest path connecting
  $u$ to $C$. By \cref{lem:exp-diam}, we know that
  $|p| \le \cdiam[\beta] \log n$.  Apply \cref{lem:exp-paths} to
  $S = C$, $T = \set{u}$, $p_u = p$, and $U = A$. We conclude that
  there is a $v \in A$ and two vertex-disjoint paths $p_u', p_v'$
  connecting $u$ and $v$ to $C$, of total length at most
  $(15\cdiam[\beta/2]/\beta)\log n$.  We can join these paths into a
  path between $u$ and $v$ by walking along $C$ in either of the two
  directions. Since $C$ has odd length this results in one odd and one
  even length path connecting $u$ to $v$, each of length at most
  $\ell + (15\cdiam[\beta/2]/\beta)\log n$, as required.
\end{proof}

%\begin{proof}[Proof of \cref{lem:exp-paths}]
\begin{proof}[Proof of Lemma~\ref{lem:exp-paths}]
  Denote by $\calP = \setdescr{p_t}{t \in T}$ a set of
  pairwise vertex-disjoint
  paths of smallest total
  length, where the path $p_t$ connects $t$ to $S$.
  Let $V(\calP) = \cup_{p \in \calP} V(p)$ denote all the vertices in the paths in $\calP$.
  Clearly, $|V(\calP)| = \totallength{T}{S} + |T|$.
  Set
  $m = |V(\calP)|(1 + 2/\beta)$
  and
  $r = \ceil{\frac{\log m}{\log(1 + \beta/2)}}$.
  Note that $n \ge |U| \ge m$ and hence
  $r \le \frac{1}{2} \cdiam[\beta/2] \log n$.

  If $|S| \ge m$, apply
  \cref{cor:exp-diam-remove} to $V(\calP)$,
  $S \setminus  V(\calP)$ and $U \setminus V(\calP)$
  to conclude that
  there is a path $p$ of length $\cdiam[\beta/2] \log n$
  connecting $S \setminus  V(\calP)$ to $U \setminus V(\calP)$ in
  $G \setminus V(\calP)$. The set $\calP \cup \set{p}$ clearly
  satisfies the conclusion of the lemma.

  Otherwise, if $|S| < m$,
  we want to get into a position where
  we can again apply \cref{cor:exp-diam-remove}. To this end, we
  define a sequence of sets of vertices
  $
  S =
  S_0 \subseteq S_1 \subseteq \ldots
  \subseteq S_\ell \subseteq V(G)
  $
  that are in some sense well-connected to $S$. We formalize this
  property after explaining how to obtain these sets.

  The set $S_{i+1}$ is defined in terms of $S_i$ using the following process.
  Let $w^i_t$ be the last vertex on the path $p_t$
  (viewed as a path from $S$ to $t$) that is in $S_i$ and
  $W_i = \setdescr{w^i_t}{t \in T}$.
  Suppose $|S_i| < m$ and there
  is a path of length at most $r$ connecting
  $S_i \setminus W_i$ to $V(\calP) \setminus S_i$
  in the graph
  $G \setminus W_i$ .
  Denote by $q_i$ a
  minimal such path,
  % $s_i \in S_i \setminus W_i$ and
  denote by
  $w$
  the endpoint of $q_i$ in $V(\calP) \setminus S_i$,
  and let
  $t_i \in T$ be such that $w \in V(p_{t_i})$.
  Then, define
  $S_{i+1} = S_i \cup q_i \cup p_{t_i}[w^{t_i}_i, w]$.
  Otherwise, if $|S_i| \ge m$ or there is no such~$q_i$,
  set $\ell = i$ and stop the process. There is an illustration
  of this process in \cref{fig:sstar-cases}.

  The following claim formalizes the well-connectedness property of
  $S_\ell$.

  \newcommand{\sstar}{s^\star}
  \begin{claim}\label{clm:len-path}
    For every vertex $\sstar \in S_\ell \setminus W_\ell$ it holds
    that
    $\totallength[{G[S_\ell \cup V(\calP)]}]{T \cup \set{\sstar}}{S}
    \le \totallength[G]{T}{S} + |S_\ell|$
    and furthermore the paths achieving this bound are the same as the
    paths in $\calP$ outside $S_\ell \setminus W_\ell$.
  \end{claim}

  \begin{proof}
    Proof by induction on $i \in \set{0, \ldots, \ell}$.
    The base case $i=0$ clearly holds -- we have for all
    $\sstar \in S_0 \setminus W_0 \subseteq S$ that
    $\totallength[{G[S_0 \cup V(\calP)]}]{T \cup \set{\sstar}}{S} =
    \totallength[G]{T}{S}$.

    Suppose the statement is true for some
    $i \in \{0, \ldots \ell-1\}$ and let us prove that is then true
    for $i+1$ as well.
    By the inductive hypothesis,
    $\totallength[{G[S_{i} \cup V(\calP)]}]{T \cup \set{s_i}}{S} \le
    \totallength[G]{T}{S} + |S_i|$, and this bound can be achieved by
    a set of paths $\calP'$ which follow $\calP$ outside $S_i \setminus W_i$.

    Fix an arbitrary $\sstar \in S_{i+1} \setminus W_{i+1}$. By the
    induction hypothesis the claim holds for
    $\sstar \in S_{i} \setminus W_i$, so we may assume\footnote{Here we are using that
      $W_{i} =
      (W_{i-1} \setminus \set{w_{i-1}^{t_i}}) \cup \set{w_i^{t_i}}$.
    }
    that either
    $\sstar \in q_i$, or
    $\sstar \in p_{t_i}[w^{t_i}_i, w^{t_i}_{i+1}]$.  If
    $\sstar \in q_i$ (excluding its endpoint~$w^{t_i}_{i+1}$) then we
    simply extend the path in $\calP'$ ending in $s_i$ with the
    subpath of $q_i$ from~$\sstar$ to~$s_i$, increasing
    the total length of $\calP'$ by
    at most $|q_i|$.  On the other hand if
    $\sstar \in p_{t_i}[w^{t_i}_i, w^{t_i}_{i+1}]$ then we reroute the
    path from $t_i$ in $\calP'$ to $s_i$ via $q_i$ and then use the
    now unused part of $p_{t_i}$ to connect $\sstar$ to $S$, again
    increasing the total length of $\calP'$ by at most $|q_i|$. There
    is an illustration of the two cases in \cref{fig:sstar-cases}.

    \begin{figure}
      \begin{center}
        \subfloat{
          \includegraphics[page=2,width=0.3\linewidth,height=\textheight,keepaspectratio]{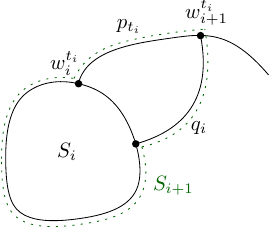}
        }\hfill
        \subfloat{
          \includegraphics[page=3,width=0.3\linewidth,height=\textheight,keepaspectratio]{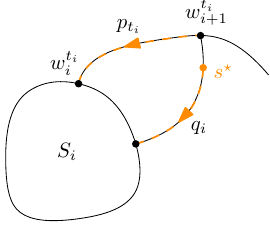}
        }\hfill
        \subfloat{
          \includegraphics[page=4,width=0.3\linewidth,height=\textheight,keepaspectratio]{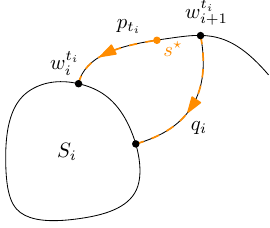}
        }\hfill
        \caption{Given the set $S_i$, the first figure depicts the
          process of obtaining the set $S_{i+1}$.
          The following figures
          indicate how to route the paths, as in the proof of
          Claim~\ref{clm:len-path},
          depending on where $\sstar$ is located.}
        \label{fig:sstar-cases}
      \end{center}
    \end{figure}

    In either case, we can connect $T$ and $\sstar$ to $S$ via vertex-disjoint paths of length at most
    \[
      \totallength[{G[S_{i+1} \cup V(\calP)]}]{T \cup \set{\sstar}}{S} \le
      \totallength[{G}]{T}{S} + |S_i| + |q_i| \le
      \totallength[{G}]{T}{S} + |S_{i+1}|,
    \]
    as desired.
  \end{proof}

  \newcommand{\ustar}{u^{\star}}
  \newcommand{\pstar}{p^{\star}}
  It is easy to see that $|S_\ell| \le 2 m + r$:
  the number of
  vertices added by
  $p_{t_i}[w^{t_i}_i, w^{t_i}_{i+1}]$ is always upper bounded by
  $|V(\calP)| \le m$.
  Suppose there is a path $\pstar$ of length
  $|p^\star| \le \cdiam[\beta/2] \log n + r$ connecting
  some vertex $\sstar \in S_\ell \setminus W_\ell$ to
  $\ustar \in U \setminus V(\calP_\ell)$
  in $G \setminus V(\calP_\ell)$. We can then ``compose'' the paths to
  conclude that
  \begin{align}
    \totallength[G]{T \cup \set{\ustar}}{S}
    &\le
      \totallength[{G[S_\ell]}]{W_\ell \cup \set{\sstar}}{S} +
      \totallength
      [G \setminus (S_\ell \setminus (W_\ell \cup \set{\sstar}))]
      {W_\ell \cup \set{\sstar}}{T \cup \set{\ustar}}\\
    &\le
      |S_\ell| + \totallength[G]{T}{S} + |\pstar|\\
    &\le
      |S_\ell| + \totallength[G]{T}{S} + \cdiam[\beta/2] \log n + r\\
    &\le
      2m + r + \totallength[G]{T}{S} + \cdiam[\beta/2] \log n + r \\
    &\le
      7(\totallength[G]{T}{S} + |T|)/\beta + 2 \cdiam[\beta/2] \log n \eqcomma
  \end{align}
  as claimed in the statement.

  It remains to establish that such a path $\pstar$ exists.
  If $|S_\ell| \ge m$, apply \cref{cor:exp-diam-remove}
  to $V(\calP)$, $S_\ell \setminus W_\ell$ and $U \setminus V(\calP)$
  to conclude that there is a path $\pstar$ of length at most
  $|\pstar| \le \cdiam[\beta/2] \log n$ that connects
  $S_\ell \setminus W_\ell$ to $U \setminus V(\calP)$ in
  $G \setminus V(\calP)$.

  Otherwise, by construction, $S_\ell \setminus W_\ell$ cannot reach
  $V(\calP) \setminus W_\ell$ within $r$ steps in
  $G \setminus W_\ell$.  Hence, to argue that in
  $G \setminus V(\calP)$ the ball of radius $r$ around
  $S_\ell \setminus W_\ell$ is large, we do not need to apply
  \cref{lem:exp-ball-remove} to $V(\calP)$ and
  $S_\ell \setminus W_\ell$ but in fact can apply it to $W_\ell$ and
  $S_\ell \setminus W_\ell$, where we use that
  $|S_\ell \setminus W_\ell| \ge |S| - |T| \ge 2|T|/\beta$.  This
  enables us to grow $S_\ell \setminus W_\ell$ into a set
  $\sgood = B_r^{G \setminus W_\ell}(S_\ell \setminus W_\ell) = B_r^{G
    \setminus V(\calP)}(S_\ell \setminus W_\ell)$ of size at least
  $m$.  Now we are in a position to apply \cref{cor:exp-diam-remove}
  to $V(\calP)$, $\sgood$ and $U \setminus V(\calP)$ to conclude that
  there is a path of length at most $\cdiam[\beta/2] \log n$ that
  connects $\sgood$ to $U \setminus V(\calP)$ in
  $G \setminus V(\calP)$. Taking an additional $r$ steps in
  $G[\sgood]$, one can reach $S_\ell \setminus W_\ell$, as
  required. This concludes the proof of the lemma.
\end{proof}

\section{Concluding Remarks}
\label{sec:conclusion}

We have established average-case lower bounds for refuting the perfect
matching formula and more generally the $\Card(G, \vec{t})$ formula in
random $d$-regular graphs on an odd number of vertices.
Let us conclude by discussing some further loose ends and mention some open problems.

\subsection{Polynomial Calculus Space Lower Bounds}
\label{sec:space}

The space of a PC refutation $\pi$ is the amount of memory needed to
verify $\pi$. The PC space of a formula $\calF$ is then the minimum space
required for any PC refutation $\pi$ of $\calF$. As this is rather tangential to the
rest of the paper we refer to
\cite{FLMNV13TowardsUnderstandingPC} for formal definitions.
For convenience, let us restate our result on PC space.

\SpaceTheorem*

The proof idea is to take the worst-case Tseitin lower bounds from
Filmus et al.~\cite{FLMNV13TowardsUnderstandingPC}
for which PC requires
$\Omega(n)$ space and embed these into a vertex expander of large
enough average degree. The only compication that arises is that these
formulas are defined over multigraphs --
the multigraph $H$ is obtained from
an appropriate\footnote{See the proof of Theorem 8 in
\cite{FLMNV13TowardsUnderstandingPC}.} constant degree graph $G$ by doubling each
edge. An inspection of the proof of \cref{thm:odd-minor} reveals
that $H$ may be a multigraph and we can thus implement our proof
strategy.

\begin{proof}[Proof Sketch]
  Consider the worst-case instance $H$ from Filmus et
  al.~\cite{FLMNV13TowardsUnderstandingPC} on $\eps n /\log n$
  vertices, for some small enough $\eps > 0$.
  Apply \cref{thm:odd-minor} to $H$ and $G$.
  This gives a topological embedding of $H$ in
  $G$, with no control of the parities of the length of the paths.
  Consider a restriction $\rho$ that sets the
  variables outside the embedding of $H$ such that no axiom is falsified
  (see, e.g., \cite{pitassi16frege}). By appropriately substituting
  the variables on each path of the topological embedding we
  obtain that the worst-case instance $\tsf{H}$ is an affine restriction
  of $\tsf{G}$. As an affine restriction
  only reduces
  the amount of space needed to verify a proof, we see that $\tsf{G}$
  requires PC space $\Omega(n/\log n)$.
\end{proof}

\subsection{Paths in Expanders}

The arguments used in the proof of \cref{thm:odd-minor} can be adapted
to make partial progress on a question by Friedman and
Krivelevich~\cite{fm2019cycle}. They asked, given a
positive integer $q$, whether it is possible to guarantee the
existence of a cycle whose length is divisible by $q$
in every $\alpha$-expander.

We can show that for all primes $q$ satisfying
$1/\poly(\alpha) \ll q \ll \sqrt{n/\log n} $, this indeed
holds. In fact, for all $a \in \Z_q$, we can
show that there is a cycle of length $a \bmod q$.

The idea is to embed a cycle $C_{q^2}$ of length $q^2$ into $G$ such
that between any two vertices there are two paths whose length
difference is non-zero modulo $q$. If we can ensure this, as all
$0 \neq b \in \Z_q$ are generators, we can choose one path between all
embedded vertices such that
the length of the cycle is $a \bmod q$ for any $a \in \Z_q$.

In order to obtain paths of different length modulo $q$, let us embed a
cycle $c_e$ (of length $\gg 1/\poly(\alpha)$) for each edge
$e = \set{u, v}$.
We then want to
connect the vertex embeddings $B_u, B_v$
to $c_e$ such that the two
resulting paths are of different length modulo $q$. Note that once a
vertex is connected to the cycle, there are only about $2/q$
vertices in $c_e$ such that both paths are of equal length
modulo $q$. As $q$ is rather large and thus there are few such ``bad''
vertices, when an
edge embedding has to be moved to the sets
$A, A'$, we can ensure that the set $A'$ remains relatively small
compared to $A$.

\subsection{Open Problems}

The main concrete problem left open is to reduce the degree of the
hard graphs: the embedding approach taken in the worst-case to
average-case reduction results in very large degree $d$; while
\cref{thm:pm avg lower bound} does not give an explicit estimate on
$d_0$ one can trace through the proofs and get an estimate of
around $15\,000$.  The main bottleneck that prevents us from reducing
this is the \hyperref[lem:partition]{Partition Lemma} and in
particular the dependence of $d_0$ on $c$ and $\epsilon$ in
\cref{lem:split}.  If this part could be significantly improved or
circumvented we believe that the degree of the graph could be
significantly reduced, although it would still be relatively large (at
least a few hundreds).  It would be interesting to see what happens for
very small degrees, e.g., $4$-regular graph (recall that since $n$
is odd, $d$ must be even) -- is $\pmf{G}$ hard with high probability
even for these graphs?

Another interesting question is the proof complexity of perfect
matching in Polynomial Calculus over $\F_2$ (or any other field of
characteristic $2$).  While PC$_{\F_2}$ can refute the perfect
matching formula on an odd number of vertices for parity reasons, the
situation is less clear when the number of vertices is even. Are there
graphs $G$ that do not admit perfect matchings but PC$_{\F_2}$ requires
exponential size refutations?

We establish
\cref{thm:pm avg lower bound}
for random graphs and not for
all spectral expanders.  Our proof mostly uses the expansion
properties of random graphs, except for two places: in the
\hyperref[lem:partition]{Partition Lemma}
to argue that every subgraph contains a not too short
odd cycle, and for the contiguity argument in
\cref{sec:avg-general}. However, we believe that it should be possible
to circumvent these uses of randomness and establish the lower bounds for
all spectral expanders.

\cref{thm:pm avg lower bound} only gives lower bounds for $\Card(G,
b)$ when $b = \vec{t}$ is a constant vector (and~$G$ is regular).  It
would be nice to characterize more generally for which vectors $b$ the
formula is hard.  In the analogous setting for Tseitin formulas, the
precise charges of the vertices do not matter, as long as the sum of
charges is odd the formula remains hard to refute on a random graph
\cite{BGIP01LinearGaps,gri01xor}. In the $\Card(G, b)$ case however
this is not the case.  For instance, if the vector of target degrees
$b$ violates any of the \emph{inequalities} of the Erdős-Gallai characterization of
degree sequences then SoS (in fact even Sherali-Adams) can easily
refute $\Card(G, b)$.%
\footnote{To see this, note that for any $S \subseteq V(G)$, SoS of
  degree $1$ can derive
  $\sum_{v \in S} b_v = \sum_{v \in S} \sum_{e \ni v} x_e \le |E(S,
  S)| + \sum_{e \in E(S, \overline{S})} x_e \le |S|(|S|-1) + \sum_{v
    \not \in S} \sum_{e \in E(S, v)} x_e$, and also that
  $\sum_{e \in E(S,v)} x_e \le \min(b_v, |S|)$ for every
  $v \not \in S$.  Combining these, Sherali-Adams can derive
  $\sum_{v \in S} b_v \le |S|(|S|-1) + \sum_{v \not \in S} \min(b_v,
  |S|)$ which in particular means it can derive all the inequalities
  of the Erdős-Gallai theorem.}

In the
case when $G$ is the complete graph this in fact gives a complete
characterization of the easy and hard vectors $b$ but for sparse
graphs the situation is less clear.  Is there a nice characterization
of vectors $b$ for which $\Card(G, b)$ is hard for SoS with high
probability over a random $d$-regular $G$?

More broadly, another open problem is to prove SoS lower bounds for
random CSPs that do not support pairwise uniform distributions
(cf.~the brief discussion on CSPs in \cref{sec:related work}).
Viewed this way, our results establish hardness of random monotone
$1$-in-$k$-SAT instances with two occurrences per variable, for some
large constant $k$.  Reducing $k$ corresponds to the aforementioned
problem of reducing the degree, but some other natural questions are
to look at other CSPs such as $1$-in-$k$-SAT with negated literals, or
to understand the hardness as a function of the density of the
instances.

\paragraph{Acknowledgements.}
The authors are grateful to Susanna de Rezende, Johan Håstad, Jakob
Nordström, Dmitry Sokolov and Aleksa Stanković for helpful
discussions.  In particular we would like to thank Jakob Nordström who
suggested to consider the even coloring formula and brought the
problem about Polynomial Calculus space to our attention.  Finally we
would like to thank the anonymous conference and jounral referees for
their very detailed reviews and many helpful remarks.

%Uri \bibliographystyle{alpha}
%Uri \bibliography{references}
\newpage
\printbibliography

@book{alonspencer,
author = {Alon, Noga and Spencer, Joel H.},
title = {The Probabilistic Method},
year = {2000},
isbn = {1119061954},
publisher = {Wiley Publishing},
edition = {2nd}
}

@article{ABRW04PRG,
author = {Alekhnovich, Michael and Ben-Sasson, Eli and Razborov, Alexander A. and Wigderson, Avi},
title = {Pseudorandom Generators in Propositional Proof Complexity},
journal = {SIAM Journal on Computing},
volume = {34},
number = {1},
pages = {67-88},
year = {2004},
doi = {10.1137/S0097539701389944},
}

@article{AGK20RefuteSemi,
  author    = {Jackson Abascal and
               Venkatesan Guruswami and
               Pravesh K. Kothari},
  title     = {Strongly refuting all semi-random Boolean {CSPs}},
  journal   = {CoRR},
  volume    = {abs/2009.08032},
  year      = {2020},
  url       = {https://arxiv.org/abs/2009.08032},
  archivePrefix = {arXiv},
  eprint    = {2009.08032},
  bibsource = {dblp computer science bibliography, https://dblp.org}
}

@article{AR01nonbin,
  author  = {Michael Alekhnovich and Alexander A. Razborov},
  title   = {Lower Bounds for Polynomial Calculus: {N}on-Binomial Case},
  year    = {2003},
  pages   = {18--35},
  volume  = {242},
  journal = {Proceedings of the Steklov Institute of Mathematics},
  url = {http://people.cs.uchicago.edu/~razborov/files/misha.pdf},
  note    = {Preliminary version in \emph{FOCS~'01}.},
}

@InProceedings{AH18SosTradeoff,
  author =	{Albert Atserias and Tuomas Hakoniemi},
  title =	{{Size-Degree Trade-Offs for Sums-of-Squares and Positivstellensatz Proofs}},
  booktitle =	{34th Computational Complexity Conference (CCC 2019)},
  pages =	{24:1--24:20},
  series =	{Leibniz International Proceedings in Informatics (LIPIcs)},
  ISBN =	{978-3-95977-116-0},
  ISSN =	{1868-8969},
  year =	{2019},
  volume =	{137},
  editor =	{Amir Shpilka},
  publisher =	{Schloss Dagstuhl--Leibniz-Zentrum fuer Informatik},
  address =	{Dagstuhl, Germany},
  URL =		{http://drops.dagstuhl.de/opus/volltexte/2019/10846},
  URN =		{urn:nbn:de:0030-drops-108464},
  doi =		{10.4230/LIPIcs.CCC.2019.24}
}

@article{Berge57,
 ISSN = {00278424},
 URL = {http://www.jstor.org/stable/89875},
 author = {Claude Berge},
 journal = {Proceedings of the National Academy of Sciences of the United States of America},
 number = {9},
 pages = {842--844},
 publisher = {National Academy of Sciences},
 title = {Two Theorems in Graph Theory},
 urldate = {2022-07-20},
 volume = {43},
 year = {1957}
}

@inproceedings{Berkholz18Relation,
  author    = {Christoph Berkholz},
  title     = {The Relation between Polynomial Calculus,
               {S}herali-{A}dams, and Sum-of-Squares Proofs},
  year      = {2018},
  month     = feb,
  booktitle = {Proceedings of the 35th Symposium on
              Theoretical Aspects of Computer Science ({STACS}~'18)},
  pages     = {11:1--11:14},
  series    = {Leibniz International Proceedings in Informatics (LIPIcs)},
  volume    = {96},
}

@inproceedings{BFSU96,
author = {Broder, Andrei Z. and Frieze, Alan M. and Suen, Stephen and Upfal, Eli},
title = {An Efficient Algorithm for the Vertex-Disjoint Paths Problem in Random Graphs},
year = {1996},
isbn = {0898713668},
publisher = {Society for Industrial and Applied Mathematics},
address = {USA},
booktitle = {Proceedings of the Seventh Annual ACM-SIAM Symposium on Discrete Algorithms},
pages = {261–268},
numpages = {8},
location = {Atlanta, Georgia, USA},
series = {SODA ’96}
}

@article{BGIP01LinearGaps,
  author    = {Samuel R. Buss and Dima Grigoriev and
               Russell Impagliazzo and Toniann Pitassi},
  title     = {Linear Gaps between Degrees for the Polynomial Calculus
               Modulo Distinct Primes},
  journal   = {Journal of Computer and System Sciences},
  volume    = {62},
  number    = {2},
  year      = {2001},
  month     = mar,
  pages     = {267--289},
  note      = {Preliminary version in \emph{CCC~'99}},
}

@article{BHKKMP16clique,
author = {Barak, Boaz and Hopkins, Samuel and Kelner, Jonathan and Kothari, Pravesh K. and Moitra, Ankur and Potechin, Aaron},
title = {A Nearly Tight Sum-of-Squares Lower Bound for the Planted Clique Problem},
journal = {SIAM Journal on Computing},
volume = {48},
number = {2},
pages = {687-735},
year = {2019},
doi = {10.1137/17M1138236},
URL = {https://doi.org/10.1137/17M1138236},
eprint = {https://doi.org/10.1137/17M1138236},
}

@phdthesis{Blake37Thesis,
  author = {Archie Blake},
  title  = {Canonical Expressions in {B}oolean Algebra},
  school = {University of Chicago},
  year   = {1937},
}

@incollection{BN20ProofCplx,
  author = {Samuel R. Buss and Jakob Nordström},
  title = {Proof Complexity and {SAT} Solving},
  editor = {Biere, Armin and Heule, Marijn J. H. and van Maaren, Hans and Walsh, Toby},
  booktitle = {Handbook of Satisfiability},
  edition = {2nd},
  publisher = {IOS Press},
  year = {2021},
  month = feb,
  series = {Frontiers in Artificial Intelligence and Applications},
  volume = {336},
  pages = {233--350},
  chapter = {7},
}

@book{bollobas2001randgraphs,
place={Cambridge},
edition={2},
series={Cambridge Studies in Advanced Mathematics},
title={Random Graphs},
DOI={10.1017/CBO9780511814068},
publisher={Cambridge University Press},
author={Bollobás, B{\'e}la},
year={2001},
collection={Cambridge Studies in Advanced Mathematics},
}

@ARTICLE{BBHPRRWZ17,
  title     = "The matching problem has no small symmetric {SDP}",
  author    = "Braun, G{\'a}bor and Brown-Cohen, Jonah and Huq, Arefin and
               Pokutta, Sebastian and Raghavendra, Prasad and Roy, Aurko and
               Weitz, Benjamin and Zink, Daniel",
  journal   = "Math. Program.",
  publisher = "Springer Science and Business Media LLC",
  volume    =  165,
  number    =  2,
  pages     = "643--662",
  month     =  oct,
  year      =  2017,
  language  = "en"
}

@article{BROUWER2005155,
title = "Eigenvalues and perfect matchings",
journal = "Linear Algebra and its Applications",
volume = "395",
pages = "155 - 162",
year = "2005",
issn = "0024-3795",
doi = "https://doi.org/10.1016/j.laa.2004.08.014",
url = "http://www.sciencedirect.com/science/article/pii/S0024379504003453",
author = "Andries E. Brouwer and Willem H. Haemers",
}

@inproceedings{BS14SumsOfSquaresICM,
  author    = {Boaz Barak and David Steurer},
  title     = {Sum-of-Squares Proofs and the Quest Toward
               Optimal Algorithms},
  booktitle = {Proceedings of the International Congress of
               Mathematicians (ICM)},
  year      = {2014},
  month     = aug,
  pages     = {509--533},
  volume    = {IV},
  url       = {http://www.icm2014.org/download/Proceedings_Volume_IV.pdf},
}

@inproceedings{CEI96Groebner,
  title     = {Using the {Groebner} Basis Algorithm to Find Proofs of
               Unsatisfiability},
  author    = {Matthew Clegg and Jeffery Edmonds and Russell Impagliazzo},
  pages     = {174--183},
  booktitle = {Proceedings of the 28th Annual {ACM} Symposium on
               Theory of Computing ({STOC}~'96)},
  month     = may,
  year      = {1996},
}

@article{CLRS16,
author = {Chan, Siu On and Lee, James R. and Raghavendra, Prasad and Steurer, David},
title = {Approximate Constraint Satisfaction Requires Large {LP} Relaxations},
year = {2016},
issue_date = {November 2016},
publisher = {Association for Computing Machinery},
address = {New York, NY, USA},
volume = {63},
number = {4},
issn = {0004-5411},
url = {https://doi.org/10.1145/2811255},
doi = {10.1145/2811255},
journal = {J. ACM},
month = oct,
articleno = {34},
numpages = {22}
}

@article{cook79efficiency,
 ISSN = {00224812},
 URL = {http://www.jstor.org/stable/2273702},
 author = {Stephen A. Cook and Robert A. Reckhow},
 journal = {The Journal of Symbolic Logic},
 number = {1},
 pages = {36--50},
 publisher = {[Association for Symbolic Logic, Cambridge University Press]},
 title = {The Relative Efficiency of Propositional Proof Systems},
 volume = {44},
 year = {1979}
}

@misc{CR2019large,
  doi = {10.48550/ARXIV.1901.09349},
  url = {https://arxiv.org/abs/1901.09349},
  author = {Chuzhoy, Julia and Nimavat, Rachit},
  title = {Large Minors in Expanders},
  publisher = {arXiv},
  year = {2019},
  copyright = {arXiv.org perpetual, non-exclusive license}
}

@inproceedings{DKN20Embed,
title = "Rolling backwards can move you forward: On embedding problems in sparse expanders",
author = "Nemanja Draganić and Michael Krivelevich and Rajko Nenadov",
year = "2021",
series = "Proceedings of the Annual ACM-SIAM Symposium on Discrete Algorithms",
publisher = "Association for Computing Machinery",
pages = "123--134",
editor = "Daniel Marx",
booktitle = "ACM-SIAM Symposium on Discrete Algorithms, SODA 2021",
}

@inproceedings{DMOSS19,
  author = {Deshpande, Yash and Montanari, Andrea and O’Donnell, Ryan and Schramm, Tselil and Sen, Subhabrata},
  title = {The Threshold for {SDP}-Refutation of Random Regular {NAE-3SAT}},
  year = {2019},
  publisher = {Society for Industrial and Applied Mathematics},
  address = {USA},
  booktitle = {Proceedings of the Thirtieth Annual ACM-SIAM Symposium on Discrete Algorithms},
  pages = {2305–2321},
  numpages = {17},
  location = {San Diego, California},
  series = {SODA ’19}
}

@article{edmonds65,
title={Paths, Trees, and Flowers},
volume={17},
DOI={10.4153/CJM-1965-045-4},
journal={Canadian Journal of Mathematics},
publisher={Cambridge University Press},
author={Edmonds, Jack},
year={1965},
pages={449–467}}

@inproceedings{FLMNV13TowardsUnderstandingPC,
  author    = {Yuval Filmus and Massimo Lauria and Mladen Mik\v{s}a
               and Jakob Nordström and Marc Vinyals},
  title     = {Towards an Understanding of Polynomial Calculus:
               New Separations and Lower Bounds ({E}xtended Abstract)},
  year      = {2013},
  month     = jul,
  booktitle = {Proceedings of the 40th International Colloquium on Automata,
               Languages and Programming ({ICALP}~'13)},
  series    = {Lecture Notes in Computer Science},
  publisher = {Springer},
  pages     =  {437--448},
  volume    =  {7965},
}

@article{fm2019cycle,
  doi = {10.1007/s00493-020-4434-0},
  url = {https://doi.org/10.1007/s00493-020-4434-0},
  year = {2020},
  month = nov,
  publisher = {Springer Science and Business Media {LLC}},
  volume = {41},
  number = {1},
  pages = {53--74},
  author = {Limor Friedman and Michael Krivelevich},
  title = {Cycle Lengths in Expanding Graphs},
  journal = {Combinatorica}
}

@book{friedman2004proof,
  address = {Providence, R.I.},
  author = {Friedman, Joel},
  isbn = {9780821842805 0821842803},
  publisher = {American Mathematical Society},
  refid = {227328566},
  title = {A Proof of {Alon}'s Second Eigenvalue Conjecture and Related Problems},
  year = {2008},
}

@article{GGRSV09OddMinor,
title = {On the odd-minor variant of {Hadwiger's} conjecture},
journal = {Journal of Combinatorial Theory, Series B},
volume = {99},
number = {1},
pages = {20-29},
year = {2009},
issn = {0095-8956},
doi = {https://doi.org/10.1016/j.jctb.2008.03.006},
url = {https://www.sciencedirect.com/science/article/pii/S0095895608000488},
author = {Jim Geelen and Bert Gerards and Bruce Reed and Paul Seymour and Adrian Vetta},
}

@InProceedings{GHA02proofs,
author="Grigoriev, Dima and Hirsch, Edward A. and Pasechnik, Dmitrii V.",
title="Complexity of Semi-algebraic Proofs",
booktitle="STACS 2002",
year="2002",
publisher="Springer Berlin Heidelberg",
address="Berlin, Heidelberg",
pages="419--430",
}

@InProceedings{GI19,
author="Glinskih, Ludmila
and Itsykson, Dmitry",
editor="van Bevern, Ren{\'e}
and Kucherov, Gregory",
title="{On {Tseitin} Formulas, Read-Once Branching Programs and Treewidth}",
booktitle="Computer Science -- Theory and Applications",
year="2019",
publisher="Springer International Publishing",
address="Cham",
pages="143--155",
isbn="978-3-030-19955-5"
}

@InProceedings{GIRS19FregeTseitinGeneral,
  author =	{Nicola Galesi and Dmitry Itsykson and Artur Riazanov and Anastasia Sofronova},
  title =	{{Bounded-Depth {Frege} Complexity of {Tseitin} Formulas for All Graphs}},
  booktitle =	{44th International Symposium on Mathematical Foundations of Computer Science (MFCS 2019)},
  pages =	{49:1--15},
  ISBN =	{978-3-95977-117-7},
  ISSN =	{1868-8969},
  year =	{2019},
  URL =		{http://drops.dagstuhl.de/opus/volltexte/2019/10993},
  doi =		{10.4230/LIPIcs.MFCS.2019.49},
}

@inproceedings{GKT19SpaceWidth,
  author    = {Nicola Galesi and Leszek Kołodziejczyk and Neil Thapen},
  title     = {Polynomial Calculus Space and Resolution Width},
  booktitle = {Proceedings of the 60th Annual {IEEE} Symposium on
               Foundations of Computer Science ({FOCS}~'19)},
  year      = {2019},
  month     = nov,
  pages     = {1325-1337},
}

@article{GL10Optimality,
  author    = {Nicola Galesi and Massimo Lauria},
  title     = {Optimality of Size-Degree Trade-offs for Polynomial Calculus},
  journal   = {ACM Transactions on Computational Logic},
  year      = {2010},
  month     = nov,
  pages     = {4:1--4:22},
  articleno = {4},
  volume    = {12},
  number    = {1},
}

@article{gri01xor,
title = "Linear lower bound on degrees of Positivstellensatz calculus proofs for the parity",
journal = "Theoretical Computer Science",
volume = "259",
number = "1",
pages = "613 - 622",
year = "2001",
issn = "0304-3975",
doi = "https://doi.org/10.1016/S0304-3975(00)00157-2",
url = "http://www.sciencedirect.com/science/article/pii/S0304397500001572",
author = "Dima Grigoriev",
}

@article{Has20,
author = {H{\aa}stad, Johan},
title = {On Small-Depth {Frege} Proofs for {Tseitin} for Grids},
year = {2020},
issue_date = {February 2021},
publisher = {Association for Computing Machinery},
address = {New York, NY, USA},
volume = {68},
number = {1},
issn = {0004-5411},
url = {https://doi.org/10.1145/3425606},
doi = {10.1145/3425606},
journal = {J. ACM},
month = {nov},
articleno = {1},
numpages = {31}
}

@ARTICLE{Hoory06expandergraphs,
    author = {Shlomo Hoory and Nathan Linial and Avi Wigderson},
    title = {Expander graphs and their applications},
    journal = {BULL. AMER. MATH. SOC. },
    year = {2006},
    volume = {43},
    number = {4},
    pages = {439--561}
}

@article{IPS99LowerBounds,
  author    = {Russell Impagliazzo and Pavel Pudl{\'a}k and
               Ji{\v{r}}\'i Sgall},
  title     = {Lower Bounds for the Polynomial Calculus and the {G}r{\"o}bner
               Basis Algorithm},
  journal   = {Computational Complexity},
  volume    = {8},
  number    = {2},
  year      = {1999},
  pages     = {127--144},
}

@article{IRSS19,
  doi = {10.1007/s00037-021-00213-2},
  url = {https://doi.org/10.1007/s00037-021-00213-2},
  year = {2021},
  month = aug,
  publisher = {Springer Science and Business Media {LLC}},
  volume = {30},
  number = {2},
  author = {Dmitry Itsykson and Artur Riazanov and Danil Sagunov and Petr Smirnov},
  title = {Near-Optimal Lower Bounds on Regular Resolution Refutations of {Tseitin} Formulas for All Constant-Degree Graphs},
  journal = {computational complexity}
}

@book{JLR00,
  address = {New York; Chichester},
  author = {Janson, Svante and \L{}uczak, Tomasz and Ruci\'{n}ski, Andrzej},
  publisher = {John Wiley \& Sons},
  title = {Theory of random graphs},
  year = {2000},
}

@inproceedings{kr96ShortPaths,
author = {Jon M. Kleinberg and Ronitt Rubinfeld},
title = {Short Paths in Expander Graphs},
year = {1996},
publisher = {IEEE Computer Society},
address = {USA},
booktitle = {Proceedings of the 37th Annual Symposium on Foundations of Computer Science},
pages = {86},
numpages = {1},
series = {FOCS ’96}
}

@inproceedings{kmow17anycsp,
author = {Kothari, Pravesh K. and Mori, Ryuhei and O’Donnell, Ryan and Witmer, David},
title = {Sum of Squares Lower Bounds for Refuting Any {CSP}},
year = {2017},
isbn = {9781450345286},
publisher = {ACM},
address = {New York, NY, USA},
url = {https://doi.org/10.1145/3055399.3055485},
doi = {10.1145/3055399.3055485},
booktitle = {Proceedings of the 49th Annual ACM SIGACT Symposium on Theory of Computing},
pages = {132–145},
numpages = {14},
location = {Montreal, Canada},
series = {STOC 2017}
}

@article{KMR17,
author = {Kothari, Pravesh K. and Meka, Raghu and Raghavendra, Prasad},
title = {Approximating Rectangles by Juntas and Weakly Exponential Lower Bounds for {LP} Relaxations of {CSPs}},
journal = {SIAM Journal on Computing},
volume = {51},
number = {2},
pages = {STOC17-305-STOC17-332},
year = {2022},
doi = {10.1137/17M1152966},
URL = {https://doi.org/10.1137/17M1152966},
eprint = {https://doi.org/10.1137/17M1152966},
}

@book{Krajicek19ProofComplexity,
  author    = {Jan Kraj{\'i}{\v{c}}ek},
  title     = {Proof Complexity},
  year      = {2019},
  month     = mar,
  publisher = {Cambridge University Press},
  volume    = {170},
  series    = {Encyclopedia of Mathematics and Its Applications},
}

@inbook{krivelevich2018expanders,
author={Krivelevich, Michael},
place={Cambridge},
series={London Mathematical Society Lecture Note Series},
title={Expanders – how to find them, and what to find in them},
DOI={10.1017/9781108649094.005},
booktitle={Surveys in Combinatorics 2019},
publisher={Cambridge University Press},
year={2019},
pages={115–142},
collection={London Mathematical Society Lecture Note Series}
}

@article{kn2018MinorsNoCut,
  doi = {10.1093/imrn/rnz086},
  url = {https://doi.org/10.1093/imrn/rnz086},
  year = {2019},
  month = may,
  publisher = {Oxford University Press ({OUP})},
  volume = {2021},
  number = {12},
  pages = {8996--9015},
  author = {Michael Krivelevich and Rajko Nenadov},
  title = {Complete Minors in Graphs Without Sparse Cuts},
  journal = {International Mathematics Research Notices}
}

@inproceedings{lasserre2001,
  title={An explicit exact SDP relaxation for nonlinear 0-1 programs},
  author={Lasserre, Jean B.},
  booktitle={International Conference on Integer Programming and Combinatorial Optimization},
  pages={293--303},
  year={2001},
  organization={Springer}
}

@inproceedings{LRS15,
  author    = {James R. Lee and
               Prasad Raghavendra and
               David Steurer},
  title     = {Lower Bounds on the Size of Semidefinite Programming Relaxations},
  booktitle = {Proceedings of the 47th Annual {ACM} on Symposium on Theory
               of Computing, {STOC} 2015, Portland, OR, USA, June 14-17, 2015},
  pages     = {567--576},
  publisher = {{ACM}},
  year      = {2015},
  url       = {https://doi.org/10.1145/2746539.2746599},
  doi       = {10.1145/2746539.2746599},
  bibsource = {dblp computer science bibliography, https://dblp.org}
}

@article{LS91Cones,
  title     = {Cones of matrices and set-functions and $0$-$1$ optimization},
  author    = {Lov{\'a}sz, L{\'a}szl{\'o} and Schrijver, Alexander},
  journal   = {SIAM Journal on Optimization},
  volume    = {1},
  number    = {2},
  pages     = {166--190},
  year      = {1991},
}

@article{Markstrom06Locality,
  author    = {Klas Markstr{\"o}m},
  title     = {Locality and Hard {SAT}-Instances},
  journal   = {Journal on Satisfiability, Boolean Modeling and Computation},
  volume    = {2},
  number    = {1-4},
  year      = {2006},
  pages     = {221--227},
}

@article{mohar,
title = "Isoperimetric numbers of graphs",
journal = "Journal of Combinatorial Theory, Series B",
volume = "47",
number = "3",
pages = "274 - 291",
year = "1989",
issn = "0095-8956",
doi = "https://doi.org/10.1016/0095-8956(89)90029-4",
url = "http://www.sciencedirect.com/science/article/pii/0095895689900294",
author = "Bojan Mohar",
}

@InProceedings{MN15,
  author    = {Mladen Mik\v{s}a and Jakob Nordström},
  title     = {A Generalized Method for Proving Polynomial Calculus
               Degree Lower Bounds},
  booktitle = {Proceedings of the 30th Annual
               Computational Complexity Conference ({CCC}~'15)},
  year      = {2015},
  month     = jun,
  pages     = {467--487},
  series    = {Leibniz International Proceedings in Informatics (LIPIcs)},
  volume    = {33},
}

@inproceedings{MPW15SumOfSquaresPlantedClique,
  title     = {Sum-of-Squares Lower Bounds for Planted Clique},
  author    = {Raghu Meka and Aaron Potechin and Avi Wigderson},
  year      = {2015},
  month     = jun,
  booktitle = {Proceedings of the 47th Annual ACM Symposium on
               Theory of Computing ({STOC}~'15)},
  pages     = {87--96},
}

@book{mu2005,
author = {Mitzenmacher, Michael and Upfal, Eli},
title = {Probability and Computing: Randomized Algorithms and Probabilistic Analysis},
year = {2005},
isbn = {0521835402},
publisher = {Cambridge University Press},
address = {USA}
}

@phdthesis{parrilo2000,
  title={Structured semidefinite programs and semialgebraic geometry methods in robustness and optimization},
  author={Parrilo, Pablo A.},
  year={2000},
  school={California Institute of Technology}
}

@inproceedings{pitassi16frege,
author = {Pitassi, Toniann and Rossman, Benjamin and Servedio, Rocco A. and Tan, Li-Yang},
title = {Poly-Logarithmic {Frege} Depth Lower Bounds via an Expander Switching Lemma},
year = {2016},
isbn = {9781450341325},
publisher = {Association for Computing Machinery},
address = {New York, NY, USA},
url = {https://doi.org/10.1145/2897518.2897637},
doi = {10.1145/2897518.2897637},
booktitle = {Proceedings of the Forty-Eighth Annual ACM Symposium on Theory of Computing},
pages = {644–657},
numpages = {14},
location = {Cambridge, MA, USA},
series = {STOC ’16}
}

@InProceedings{potechin17GoodStory,
  author =	{Aaron Potechin},
  title =	{{Sum of Squares Lower Bounds from Symmetry and a Good Story}},
  booktitle =	{10th Innovations in Theoretical Computer Science  Conference (ITCS 2019)},
  pages =	{61:1--61:20},
  series =	{Leibniz International Proceedings in Informatics (LIPIcs)},
  ISBN =	{978-3-95977-095-8},
  ISSN =	{1868-8969},
  year =	{2018},
  volume =	{124},
  publisher =	{Schloss Dagstuhl--Leibniz-Zentrum fuer Informatik},
  address =	{Dagstuhl, Germany},
  URL =		{http://drops.dagstuhl.de/opus/volltexte/2018/10154},
  URN =		{urn:nbn:de:0030-drops-101545},
  doi =		{10.4230/LIPIcs.ITCS.2019.61}
}

@inproceedings{potechin20ordering,
  author    = {Aaron Potechin},
  title     = {Sum of Squares Bounds for the Ordering Principle},
  booktitle = {35th Computational Complexity Conference, {CCC} 2020, July 28-31,
               2020, Saarbr{\"{u}}cken, Germany (Virtual Conference)},
  series    = {LIPIcs},
  volume    = {169},
  pages     = {38:1--38:37},
  publisher = {Schloss Dagstuhl - Leibniz-Zentrum f{\"{u}}r Informatik},
  year      = {2020},
  url       = {https://doi.org/10.4230/LIPIcs.CCC.2020.38},
  doi       = {10.4230/LIPIcs.CCC.2020.38},
  bibsource = {dblp computer science bibliography, https://dblp.org}
}

@article{Razborov98,
  author    = {Alexander A. Razborov},
  title     = {Lower Bounds for the Polynomial Calculus},
  journal   = {Computational Complexity},
  volume    = {7},
  number    = {4},
  year      = {1998},
  month     = dec,
  pages     = {291--324},
}

@inproceedings{Razborov02ProofComplexityPHP,
  author    = {Alexander A. Razborov},
  title     = {Proof Complexity of Pigeonhole Principles},
  year      = {2002},
  pages     = {100--116},
  booktitle     = {5th International Conference on Developments in
               Language Theory, ({DLT}~'01), Revised Papers},
  month     = jul,
  publisher = {Springer},
  series    = {Lecture Notes in Computer Science},
  volume    = {2295},
}

@article{razborov2017width,
author = {Razborov, Alexander},
title = {On the Width of Semialgebraic Proofs and Algorithms},
year = {2017},
issue_date = {November 2017},
publisher = {INFORMS},
address = {Linthicum, MD, USA},
volume = {42},
number = {4},
issn = {0364-765X},
url = {https://doi.org/10.1287/moor.2016.0840},
doi = {10.1287/moor.2016.0840},
journal = {Math. Oper. Res.},
month = nov,
pages = {1106–1134},
numpages = {29}
}

@phdthesis{Riis93Thesis,
  author = {S\o{}ren Riis},
  title  = {Independence in Bounded Arithmetic},
  school = {University of Oxford},
  year   = {1993},
}

@article{Rothvoss17,
  author = {Rothvoss, Thomas},
  title = {The Matching Polytope Has Exponential Extension Complexity},
  year = {2017},
  issue_date = {November 2017},
  publisher = {Association for Computing Machinery},
  address = {New York, NY, USA},
  volume = {64},
  number = {6},
  issn = {0004-5411},
  url = {https://doi.org/10.1145/3127497},
  doi = {10.1145/3127497},
  journal = {J. ACM},
  month = sep,
  articleno = {41},
  numpages = {19}
}

@article{RS86,
title = {Graph minors. V. Excluding a planar graph},
journal = {Journal of Combinatorial Theory, Series B},
volume = {41},
number = {1},
pages = {92-114},
year = {1986},
issn = {0095-8956},
doi = {https://doi.org/10.1016/0095-8956(86)90030-4},
url = {https://www.sciencedirect.com/science/article/pii/0095895686900304},
author = {Neil Robertson and P.D Seymour},
}

@inproceedings{Schoenebeck08LinearLevel,
  author    = {Grant Schoenebeck},
  title     = {Linear Level {L}asserre Lower Bounds for Certain $k$-{CSPs}},
  pages     = {593--602},
  booktitle = {Proceedings of the 49th Annual {IEEE} Symposium on
               Foundations of Computer Science ({FOCS}~'08)},
  month =     oct,
  year =      {2008},
}

@Inbook{Seymour2016Hadwiger,
author="Seymour, Paul",
editor="Nash, Jr., John Forbes and Rassias, Michael Th.",
title="Hadwiger's Conjecture",
bookTitle="Open Problems in Mathematics",
year="2016",
publisher="Springer International Publishing",
address="Cham",
pages="417--437",
isbn="978-3-319-32162-2",
doi="10.1007/978-3-319-32162-2_13",
url="https://doi.org/10.1007/978-3-319-32162-2_13"
}

@article{Shor87,
  title={Class of global minimum bounds of polynomial functions},
  author={N.Z. Shor},
  journal={Cybernetics},
  year={1987},
  volume={23},
  pages={731-734}
}

@article{UF96SimplifiedLowerBounds,
  author    = {Alasdair Urquhart and Xudong Fu},
  title     = {Simplified Lower Bounds for Propositional Proofs},
  journal   = {Notre Dame Journal of Formal Logic},
  volume    = {37},
  number    = {4},
  year      = {1996},
  pages     = {523--544},
}

@article{yannakakis88,
title = {Expressing combinatorial optimization problems by Linear Programs},
journal = {Journal of Computer and System Sciences},
volume = {43},
number = {3},
pages = {441-466},
year = {1991},
issn = {0022-0000},
doi = {https://doi.org/10.1016/0022-0000(91)90024-Y},
url = {https://www.sciencedirect.com/science/article/pii/002200009190024Y},
author = {Mihalis Yannakakis},
}

%% \clearpage

\appendix

\section{Worst-Case Lower Bounds}
\label{sec:worst-case}

In this section we describe a general reduction from the Tseitin
formula to the Perfect Matching formula as it appeared in
\cite{BGIP01LinearGaps} for Polynomial Calculus. We then observe that this
reduction also works for the SoS and bounded depth Frege proof
systems.

Starting from a graph $\gts$ such that the Tseitin
formula $\tsf{\gts}$ is hard for a proof system \psys, we want to
craft a graph $\gpm$ so that $\pmf{\gpm}$ is hard for
\psys.
To simplify the presentation, let us assume that $\gts$ is
$d$-regular. As we are interested in unsatisfiable instances, i.e.,
when $\gts$ has an odd number of vertices, we may assume that $d$ is even.

The graph $\gpm$ is a ``blow-up'' (or ``lift'') of $\gts$: each vertex
in $\vts$ is lifted to a clique of $d+1$ vertices and each lifted edge
connects a single pair of vertices from the corresponding cliques.  If we denote
the lifted vertices of $v \in \vts$ by
$\lift{v} = \set{\vtex{\star}, \vtex{1}, \ldots, \vtex{d}}$,
we add for each edge $\set{u, v} \in \ets$, where $v$ is the
$i$th neighbor of $u$ and $u$ is the $j$th neighbor of $v$, an
edge $\set{\vtex[u]{i}, \vtex[v]{j}}$.
An
illustration of the construction of $\gpm$ can be found in
\cref{fig:red-ex}.
\begin{figure}
  \begin{center}
    \subfloat[The graph $\gts$.]{%
      \includegraphics[width=0.3\textwidth, page=1]{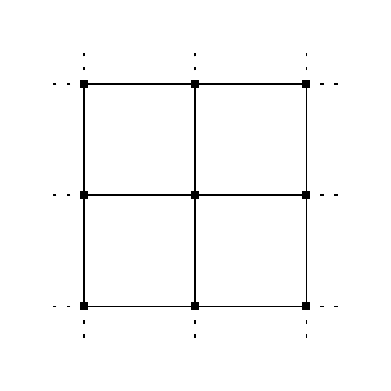}%
    }\hfill%
    \subfloat[The lift of a single vertex.]{%
      \includegraphics[width=0.3\textwidth, page=2]{fig/red_deg_3.pdf}%
    }\hfill%
    \subfloat[The lifted graph $\gpm$.]{%
      \includegraphics[width=0.3\textwidth, page=3]{fig/red_deg_3.pdf}%
    }\hfill%
    \caption{An illustration of the blow-up construction, starting from
      a 4-regular graph.}
    \label{fig:red-ex}
  \end{center}
\end{figure}

For intuition, let us describe how we would obtain a satisfying assignment to
the Perfect Matching Formula from a hypothetical satisfying assignment to the
Tseitin Formula. Set the lifted edges to the same value as they are
set to in the Tseitin Formula. Now observe that each charge is odd and
hence there is an even number of vertices left that are not matched
yet in each $\lift{v}$, for $v \in \vts$. As the vertices in
$\lift{v}$ form a clique, we can select a perfect matching on these unmatched vertices
to obtain a satisfying assignment to the Perfect Matching Formula.

Buss et al.~\cite{BGIP01LinearGaps} showed that Polynomial Calculus
can simulate this reduction.
\begin{theorem}[\cite{BGIP01LinearGaps}]\label{thm:worst-pc}
  There are graphs $G$ on an odd number of vertices $n$
  and maximum degree $\Delta(G) = 5$
  such that Polynomial Calculus over any field of characteristic
  different from $2$ requires degree $\Theta(n)$ to refute $\pmf{G}$.
\end{theorem}

What remains is to check that this reduction also gives Perfect
Matching worst-case lower bounds for Sum-of-Squares and bounded depth Frege.
This straightforward verification is carried out in the following
two sections.

\subsection{Sum-of-Squares}

A nice property of the Sum-of-Squares system is that if the variables
for a formula $\calQ$ can be expressed as well-behaved low-degree
polynomials in the variables of another formula $\calP$ for which a
pseudo-expectation exists, then a pseudo-expectation also exists for
$\calQ$. This property is well-known but let us state and quickly
prove the exact version we need.

\begin{claim}\label{claim:sos-reduction}
  Let $\calP \subseteq \R[x_1, \ldots, x_n]$
  and $\calQ \subseteq \R[y_1, \ldots, y_m]$ be two
  systems of polynomial equations.
  Let $\pedef_\calQ$ be a degree $D$ pseudo-expectation for
  $\calQ$.
  Suppose there is a function
  $f: \set{x_1, \ldots, x_m} \to \R[y_1, \ldots y_n]$,
  mapping the $x$ variables to polynomials in $y$ of degree at most $t$.
  Extend $f$ to polynomials by applying the function to each
  variable individually.
  If $f$ satisfies that
  $\pe[\calQ]{f(r \cdot p)} = 0$,
  for all $p \in \calP$ and
  $r \in \R[x_1, \ldots x_m]$ of degree $\deg(r \cdot p) \le D/t$,
  then
  $\pedef_\calQ \circ f$ is a degree $D/t$ pseudo-expectation for
  $\calP$.
\end{claim}

\begin{proof}
  As $f$ only maps variables we have that
  $\pe[\calQ]{f(1)} = \pe[\calQ]{1} = 1$.
  Also, we need to check that $\pe[\calQ]{f(s^2)} \ge 0$ for
  $s \in \R[x_1, \ldots, x_m]$ of degree $\deg(s) \le D/2t$.
  As we apply $f$ individually to each variable, we can write
  $
  \pe[\calQ]{f(s^2)} =
  \pe[\calQ]{\big(\sum_{t \in s} f(t) \big)^2} \ge 0
  $,
  as $\pedef_\calQ$ is a degree $D$ pseudo-expectation.
\end{proof}

In order to apply Claim~\ref{claim:sos-reduction} with $\calQ = \tsf{\gts}$ %Uri
and $\calP = \pmf{\gpm}$, we need to express each
variable from the Perfect Matching formula as a low degree polynomial
in the Tseitin variables.

Let us recall some notation. For a vertex $v \in \vts$, let $Y_v$
be the set of Tseitin variables corresponding to edges incident to
$v$ and denote by $A_v$ all boolean assignments to $Y_v$ that satisfy
the vertex axiom of $v$, i.e., assignments that
set an odd number of edges to true.
For a Tseitin variable $y_e$, where $e \in \ets$,
let $\lift{y_e} \in \epm$ denote the lifted
edge variable.

With this notation at hand, let us define the function $f$ to use in
Claim~\ref{claim:sos-reduction}. %Uri
Variables that correspond to lifted edges, $x_e = \lift{y_{e'}}$ for some
$e' \in \ets$, are set to 1 if and only if $y_{e'}$ is set to~$1$ and the
variables in $Y_v$ are set according to some assignment in $A_v$:
\begin{align}
  f(x_e) =
  \sum_{
    \substack{
      \alpha \in A_v \\
      \alpha(y_{e'}) = 1
    }
  }
  \ind{Y_v = \alpha} \eqperiod
\end{align}
Note that this is a polynomial of degree $\deg(v) = d$ in the $y_e$'s.
For each assignment $\alpha \in A_v$, set the variables
in $\lift{Y_v}$ according to $\alpha$ and fix a matching $m_\alpha$ on the
vertices in $\lift{v}$ not matched by $\alpha$. For any edge
$e \subseteq \lift{v}$, let
\begin{align}
  f(x_e) = \sum_{\substack{\alpha \in A_v\\ e \in m_\alpha}} \ind{Y_v
    = \alpha} \eqperiod
  %% \prod_{\alpha(y_{e'}) = 1} y_{e'} \prod_{\alpha(y_{e'}) = 0}
  %% \bar{y}_{e'} \eqperiod
\end{align}
If we apply $f$ individually to each variable, we
claim that for
$i \in \set{1, \ldots, d}$
and $v \in \vts$ the polynomial
$f(\pmax{\vtex{i}})$ is equal to the Tseitin axiom $\tsax{v}$:
\begin{align}
  f(\pmax{\vtex{i}})
  &= \sum_{e \ni \vtex{i}} f(x_e) - 1\\
  &=
  \sum_{\substack{\alpha \in A_v \\ \alpha(y_{e'})=1}}
  \ind{Y_v = \alpha} +
  \sum_{\substack{\alpha \in A_v \\ \alpha(y_{e'})=0}}
  \ind{Y_v = \alpha} - 1 \\
  &= \tsax{v} \eqcomma
\end{align}
using that the $m_\alpha$ are matchings. The axioms
$\pmax{\vtex{\star}}$ are handled similarly:
\begin{align}
    f(\pmax{\vtex{\star}})
  = \sum_{e \ni \vtex{\star}} f(x_e) - 1
  =
  \sum_{\substack{\alpha \in A_v }}
  \ind{Y_v = \alpha} - 1 
  = \tsax{v} \eqperiod
\end{align}

As $\pedef_{\tsf{\gts}}$ maps
all axioms multiplied by a low degree polynomial to 0, the same holds
for $\pedef_{\tsf{\gts}} \circ f$ and we can thus apply
Claim~\ref{claim:sos-reduction}. %Uri

We conclude that if there is a degree $D$ pseudo-expectation
$\pedef_{\tsf{\gts}}$ for the Tseitin Formula $\tsf{\gts}$, then
there is a degree $D/d$ pseudo-expectation $\pedef_{\pmf{\gpm}}$
for the Perfect Matching formula over the lifted graph $\gpm$. Using
Grigoriev's Tseitin lower bound \cite{gri01xor} we obtain the
following Theorem.

\begin{theorem}\label{thm:worst-sos}
  There are graphs $G$ on an odd number of vertices $n$
  and maximum degree $\Delta(G) = 5$ for which SoS
  requires degree $\Theta(n)$ to refute $\pmf{G}$.
\end{theorem}

\subsection{Bounded Depth Frege}

In this section we intend to prove the following theorem.

\begin{theorem}\label{thm:worst-frege}
  There is a constant $c > 0$ such that the following holds.
  Suppose $D \le \frac{c \log n}{\log \log n}$. Then
  there are graphs $G$ on an odd number of vertices $n$
  and maximum degree $\Delta(G) = 5$ such that
  any depth-$D$ Frege refutation of $\pmf{G}$ requires size
  $\exp(\Omega(n^{c/D}))$.
\end{theorem}

As in the previous section we use a function $f$, mapping Perfect
Matching variables to low depth formulas in the Tseitin Variables,
to argue that we can transform a refutation of
$\pmf{\gpm}$ into a refutation of the Tseitin formula
$\tsf{\gts}$. Assuming that this can be done, we use
the following recent result of Håstad about the Tseitin
formula over the grid to obtain \cref{thm:worst-frege}.

\begin{theorem}[\cite{Has20}]
  Suppose that $D \le \frac{\log n}{59 \log \log n}$, then
  any depth-$D$ Frege refutation of the
  Tseitin formula on the $n \times n$ grid requires size
  $\exp(\Omega(n^{1/58(D+1)}))$.
\end{theorem}

In the previous section $f$ mapped to polynomials. As we are now
working with formulas we need to translate the polynomials to
formulas. This is straightforward; reusing
notation from the previous section, let
\begin{align}
  f(x_e) =
  \bigvee_{
    \substack{
      \alpha \in A_v \\
      \alpha(y_{e'}) = 1
    }
  }
  \ind{Y_v = \alpha} \eqcomma
\end{align}
if $x_e = \lift{y_{e'}}$ is a lifted edge. Else let
\begin{align}
  f(x_e) = \bigvee_{\substack{\alpha \in A_v\\ e \in m_\alpha}}
  \ind{Y_v = \alpha} \eqperiod
\end{align}

Suppose there is a depth-$D$ Frege refutation $\pi$ of the Perfect
Matching formula $\pmf{\gpm}$. Replace each occurrence of a Perfect
Matching variable $x_e$ by $f(x_e)$ to obtain a depth-$(D+2)$
refutation $\pi'$.  We claim that $\pi'$ can be massaged into a
refutation of the Tseitin formula $\tsf{\gts}$ of size
$O_d(\Size(\pi))$.

To this end we need to argue that $f$ maps Perfect Matching axioms to
Tseitin Axioms or tautologies that are derivable in small size and
depth. Analoguous to SoS observe that for all $v \in \vts$ and
$i \in \set{1, \ldots, d}$,
\begin{align}
  f(\bigvee_{e \ni \vtex{i}} x_e)
  &= \bigvee_{e \ni \vtex{i}} f(x_e)\\
  &=
  \bigvee_{\substack{\alpha \in A_v \\ \alpha(y_{e'})=1}}
  \ind{Y_v = \alpha} \vee
  \bigvee_{\substack{\alpha \in A_v \\ \alpha(y_{e'})=0}}
  \ind{Y_v = \alpha} \eqperiod
\end{align}
Note that that the final formula is equal to the axiom $\tsax{v}$, up
to a reordering of the terms. As the axioms are over $d$ variables,
which is constant in our case, this formula can be derived from
$\tsax{v}$ in constant depth and size. The axiom for the vertex
$\vtex{\star}$ is handled in a similar manner.

Last we
need to show that the axioms $\bar{x}_e \vee \bar{x}_{e'}$, for
edges $e \neq e' \in \epm$ satisfying $e \cap e' \neq \emptyset$,
are mapped to a
tautology derivable in small size and depth. If we let
$\set{\vtex{i}} = e \cap e'$ we can write
\begin{align}
  f(\bar{x}_e \vee \bar{x}_{e'}) =
  \big( \lnot \bigvee_{\beta \in B} \ind{Y_v = \beta} \big)\vee
  \big( \lnot \bigvee_{\gamma \in C} \ind{Y_v = \gamma} \big) \eqcomma
\end{align}
for disjoint subsets $B, C \subseteq A_v$. Observe that this formula
is a tautology and defined on $d$ variables. Thus it is derivable in
constant depth and size dependent on $d$, which is constant in our
case.

\section{Embedding Algorithm}
\label{sec:pseudocode}
% \input{odd-minor-algs.tex}

% \begin{algorithm}
%   \caption{Restores $\beta$-expansion of $G[C]$.}
%   \label{alg:fix-expansion}
%   \begin{algorithmic}[1]
%     \Procedure{FixExpansion}{$G, C, A, \beta$}
%     \While{$G[C]$ is not a $\beta$-expander}
%     \State $U \gets_\text{any}$ subset of $C$ such that
%     $|U| \le |C|/2$ and
%     $|N(U, C \setminus U)| < \beta |U|$
%     \State $C \gets C \setminus U$
%     \State $A \gets A \cup U$
%     \EndWhile
%     \EndProcedure
%   \end{algorithmic}
% \end{algorithm}

%\RestyleAlgo{ruled}
% \makeatletter
% \renewcommand{\@algocf@capt@plain}{above}% formerly {bottom}
% \makeatother

\SetKwComment{tcp}{$\triangleright$\hspace*{5pt}}{}%
\SetInd{24pt}{0pt}
\renewcommand{\CommentSty}[1]{\sffamily\fontsize{12}{18}\selectfont\textcolor{@TCSdarkred}{#1}\unskip}%
% \SetAlFnt{\mathversion{sansmath}\fontsize{12}{18}\selectfont\ttfamily}
% \SetAlFnt{\mathversion{sansmath}\fontsize{12}{16}\selectfont}%\ttfamily}

\begin{algorithm}
  \SetKwFunction{FixExpansion}{FixExpansion}
  % \TitleOfAlgo{Restores $\beta$-expansion of $G[C]$.}
    \Fn{$\FixExpansion(G, C, A, \beta)$}{
  % \begin{algorithmic}[1]
    % \Procedure{FixExpansion}{$G, C, A, \beta$}
    \While{$G[C]$ is not a $\beta$-expander}
    {
      $U \gets_\text{any}$ subset of $C$ such that $|U| \le |C|/2$ and $|N(U, C \setminus U)| < \beta |U|$ \;
     $C \gets C \setminus U$ \;
     $A \gets A \cup U$ 
    }
    }
  % \end{algorithmic}
  \caption{Restores $\beta$-expansion of $G[C]$.}
  \label{alg:fix-expansion}
\end{algorithm}

\newcommand{\State}{}
\newcommand{\Statex}{\BlankLine}
\newcommand{\Try}{{\bf Try}}
\newcommand{\Catch}[1]{{\bf Catch} #1}

\begin{algorithm}
\setcounter{AlgoLine}{0}
  \SetKwInOut{Require}{Require}
  \Require{Conditions of \cref{lem:cross}.} 
  \BlankLine
  \caption{Finds an \cross{(r, s)} in an $\beta$-expander $G$ as in
  the proof of \cref{lem:cross}.}
  \label{alg:embed-vertex}
  \SetKwFunction{EmbedVertex}{EmbedVertex}

  % \begin{algorithmic}[1]
     \Fn{$\EmbedVertex(G, r, s, \beta, k)$}
    {
    
    \State $\gamma \gets \frac{\beta}{3(1 + \beta)}$ \;
    \State $s \gets \max \set{1/\gamma, s}$ \;
    \State $r' \gets (1 + 1/\gamma)r$ \;
    \State $A, \calB \gets \emptyset; C \gets V(G)$ \;
    
    \While{$|\calB| < r'$}
    {
    \State $U \gets_\text{any}$ subset of $C$ 
    such that $|U| = s$ and $G[U]$ is a single connected component \;
    \State $\calB \gets \calB \cup \set{U}$; $C \gets C \setminus U$ \;
    \State $\FixExpansion(G, C, A, \gamma)$ \;
    
    \State $\mathcal{F} \subseteq \calB$ maximal 
    such that $|\cup_{F \in \mathcal{F}} N(F, C)| < \gamma s |\mathcal{F}|$ \;
    \State $\calB \gets \calB \setminus \mathcal{F}$; $A \gets A \cup_{F \in \mathcal{F}}F$ \;
    }
    \Statex
    \State $v \gets_\text{any} C$ such that $\deg_{G[C]}(v) \ge r'$ \;
    \State $F \gets$ a transversal of $\setdescr{N(B, C)}{B \in \calB}$ \;
    \State $\setdescr{p_i}{i \in [r]} \gets$ from \cref{lem:exp-path-disj} applied
    to $G[C]$, $v$ and $F$ \;
    \State \textbf{return} $\set{v} \cup \setdescr{V(p_i) \cup B_i}{i \in [r]}$ 
    \tcp*[r]{Shrink branches appropriately}
    }
    % \EndProcedure
  % \end{algorithmic}
\end{algorithm}

\SetKwFor{ForAll}{for all}{do}{}

\begin{algorithm}
\setcounter{AlgoLine}{0}
  \caption{Remove the embedding of vertex $x$.}
  \label{alg:remove-vertex}
  \SetKwFunction{UnEmbedVertex}{UnEmbedVertex}
  % \begin{algorithmic}[1]
    \Fn{$\UnEmbedVertex(A, A', B, H, I, x)$}
    {
    \State $(v, \branches) \gets \embed{x}$ \tcp*[r]{$v$ is the center and $\branches$ are the branches of $\embed{x}$}
    \State $B \gets B \setminus \embed{x}$; $I \gets I \setminus x$ \;
    \State $W \gets \emptyset$ \;
    
    \ForAll{$e \in E(H)$ such that $x \in e$ and $e$ is embedded}
    {
    \State let $U \gets \branches$ be the branch adjacent to $\embed{e}$ \;
    \State $\branches \gets \branches \setminus U$ \;
    \State $B \gets B \setminus \embed{e}$ \;
    \State $W \gets W \cup U \cup \embed{e}$ \;
    }
    \State $A \gets A \cup \branches$ 
    \tcp*[r]{First add to $A$, then to $A'$ to maintain the invariant} 
    \State $A' \gets A' \cup W \cup \set{v}$ \;
    }
  % \end{algorithmic}
\end{algorithm}

%\setlength{\algomargin}{0pt}
% \DecMargin{-20pt}
% \IncMargin{200pt}

\begin{algorithm}
\setcounter{AlgoLine}{0}
  % \algloopdefx{Try}[0]{\textbf{try}}
  % \algloopdefx{Catch}[1]{\textbf{catch} #1}
  \caption{Embeds $H$ in an $\alpha$-expander $G$ as in the proof of
    \cref{thm:odd-minor}.}
  \label{alg:embed-graph}
  \SetKwFunction{EmbedGraph}{EmbedGraph}
  % \begin{algorithmic}[1]
    \Fn{$\EmbedGraph(H, G, \alpha)$}
    {
    % \Indm
    \State $\beta \gets \alpha/3(1 + \alpha)$ \;
    \State $A, A', B \gets \emptyset; C \gets V(G)$ \;
    \State $I \gets \emptyset$ \;
    
    \While{$I \neq V(H)$}
    {
    \State $x \gets_{\text{any}} V(H) \setminus I$ \;
    \State $B_x \gets \EmbedVertex(G[C], \deg_H(x), s, \beta,k)$ \;
    % k is not defined but let us use it anyway; it's from the theorem
    % statement. 
    \State $C \gets C \setminus B_x$; $B \gets B \cup B_x$; $I \gets I \cup x$ \;
    \State $\FixExpansion(G, C, A, \beta)$ \;
    \Statex
    \State $\bFree(K) \gets$ branches of the cross $K$ 
    that are \emph{not} used to connect to a neighbor \;
    
    \ForAll{$\set{x, y} \in E(H)$ such that $y \in I$}
    {
    \Try \;
    \qquad\State $U_z \gets_{\text{any}} \bFree(\embed{z})$ such that $|N(U_z, C)| \ge \beta |U_z|$ 
    for $z \in \set{x, y}$ \;
    \Catch{no such $U_z$ for $z \in \set{x, y}$} \;
    \qquad\State $\UnEmbedVertex(A, A', B, H, I, z)$ ; \textbf{continue} \;
    
    \State $B_{xy} \gets$ odd path from \cref{lem:odd-even-path} 
    applied to $G[C]$, $N(U_x, C)$ and $N(U_y, C)$ \;
    \State $C \gets C \setminus B_{xy}$; $B \gets B \cup B_{xy}$ \;
    \State $\FixExpansion(G, C, A, \beta)$
    }
    }
    \State \textbf{return} $B$ \;
    }
  % \end{algorithmic}
\end{algorithm}

\end{document}